%% file: main.tex
\documentclass[draft,runningheads]{llncs}
\usepackage[T1]{fontenc}
\usepackage{graphicx}
\usepackage{amsmath}
\usepackage{amssymb}
\usepackage{bussproofs}
\usepackage{listings}
\usepackage[linesnumbered,ruled,vlined]{algorithm2e} 
\usepackage[driverfallback=dvipdfm,colorlinks=true,linkcolor=black]{hyperref}
\hypersetup{
   colorlinks,%
   citecolor=blue,%
   filecolor=black,%
   linkcolor=red,%
   urlcolor=black
}

\def\+#1{\mathcal{#1}}
\def\-#1{\mathbf{#1}}
\newcommand{\iml}{\textbf{IML}}
\newcommand{\fik}{\textbf{FIK}}
\newcommand{\dik}{\textbf{LIK}}
\newcommand{\ik}{\textbf{IK}}

\newcommand{\wk}{\textbf{WK}}
\newcommand{\did}{\textbf{LIKD}}
\newcommand{\dit}{\textbf{LIKT}}
\newcommand{\cfik}{$\mathbf{C}_{\fik}$}

\newcommand{\cdid}{$\mathbf{C}_{\did}$}
\newcommand{\cdit}{$\mathbf{C}_{\dit}$}
\newcommand{\cdidd}{$\mathbf{C}_{\did^-}$}
\newcommand{\cditt}{$\mathbf{C}_{\dit^-}$}
\newcommand{\cdikk}{$\mathbf{C}_{\dik}$}
\newcommand{\local}{$\mathbf{C}_{\dik^-}$}
\renewcommand{\#}{\sharp}
\newcommand{\axiomd}{\textbf{D}}
\newcommand{\axiomt}{\textbf{T}}
\newcommand{\ditra}{\textbf{LIK4}}

\newcommand{\remove}[1]{}

\def\K{\mathbf{K}}

\def\CD{\mathbf{RV}}

\def\DP{\mathbf{DP}}
\def\NotPosBot{\mathbf{N}}
\def\NEC{\mathbf{NEC}}

\def\MP{\mathbf{MP}}

\def\At{\mathbf{At}}
\def\Fo{{\mathcal L}}
\def\fcfra{\mathbf{fc}}

\def\fdcfra{\mathbf{fdc}}

\def\IPL{\mathbf{IPL}}

\def\AxiomSer{\mathbf{D}}
\def\AxiomRef{\mathbf{T}}
\def\AxiomTra{\mathbf{4}}

\AddToHook{cmd/appendix/before}{%
  \setcounter{lemma}{0}%
  \setcounter{proposition}{0}%
  \setcounter{definition}{0}%
}

\begin{document}
\title{Local Intuitionistic Modal Logics\\and Their Calculi}
%
%
\author{
Philippe Balbiani\inst{1}\and
Han Gao\inst{2}\and
Çi\u{g}dem Gencer\inst{1}\and
Nicola Olivetti\inst{2}
}
\authorrunning{Balbiani, Gao, Gencer and Olivetti}
%
\institute{CNRS-INPT-UT3, IRIT, Toulouse, France\\
\email{\{philippe.balbiani, cigdem.gencer\}@irit.fr}
\and
Aix Marseille University, CNRS, LIS, Marseille, France\\
\email{\{gao.han, nicola.olivetti\}@lis-lab.fr}}
\maketitle              
\begin{abstract}
%
%
%
%
We investigate intuitionistic modal logics with locally interpreted $\square$ and $\lozenge$.
The basic logic \dik\ is stronger than constructive modal logic \wk\ and incomparable with intuitionistic modal logic \ik.
We propose an axiomatization of  \dik\ and some of its extensions.
We propose  bi-nested calculi for \dik\ and these extensions, thus providing both a decision procedure and a procedure of finite countermodel extraction.
\end{abstract}

\input{00-new-intro.tex}
\input{01-logic.tex}
\input{02-new-calculi.tex}
\input{03-terminationNEW.tex}
\input{04-completeness.tex}
\input{05-conclusion.tex}

%
%
%
\bibliographystyle{splncs04}
\bibliography{refs}

\appendix
\input{A-sec2-proofs.tex}
\input{A-sec3-proofs.tex}
\input{A-sec4-proofs.tex}
\input{A-sec5-proofs.tex}

\end{document}

%% file: 00-new-intro.tex
\section{Introduction}\label{sec:intro}
%
%
%
%
Along the time, two traditions emerged in Intuitionistic modal logic (\iml).
The first tradition~---~called \textit{intuitionistic modal logics}~\cite{Fischer:Servi:1977,Fischer:Servi:1978,Fischer:Servi:1984,Plotkin:Stirling:1986,Simpson:1994}~---~aims to define modalities justified by an intuitionistic meta-theory.
In this tradition  the basis logic is \textbf{IK} which is considered as the intuitionistic counterpart of the minimal normal modal logic \textbf{K}.
The second tradition~---~called \textit{constructive modal logics}~---~is mainly motivated by computer science applications (Curry-Howard correspondence, verification and contextual reasoning, etc).
In this tradition, the basic logics are \textbf{CCDL}~\cite{Wijesekera:1990} and  \textbf{CK}~\cite{Bellin:et:al:2001}.

However, there are natural logics that have received little interest and deserve to be studied.
The present work aims to study the logic \dik~(Local \ik) where, for all models $(W,{\leq},{R},V)$, the modal operators are classically interpreted:
\begin{quote}
(1) $x\Vdash \Box A$ iff  for all $y$ such that $Rxy$ it holds $y \Vdash A$; \\
(2) $x\Vdash \Diamond A$ iff  there exists $y$ such that $Rx y$  and  $y \Vdash A$.
\end{quote}
We call these forcing conditions ``local'' as they do not involve worlds $\leq$-greater than $x$. 
Meanwhile, we require that intuitionistic axioms remain valid in the full logic.
This is expressed by the {\em hereditary property} (HP), saying that for any formula $A$, if $A$ is forced by a world $x$, it will also be forced by any upper world of $x$.
In order to ensure (HP), we need to postulate, in models $(W,{\leq},{R},V)$, two frame conditions which relate $\leq$ and $R$: the conditions of \emph{downward confluence} and \emph{forward confluence}~\cite{Balbiani:et:al:2021,BozicDozen,Fischer:Servi:1984,Simpson:1994}.
%
%

%
%
Bo{\v{z}}i{\'c} and Do{\v{s}}en~\cite{BozicDozen} studied separately the $\Box$ fragment and the $\Diamond$ fragment of \dik.
They also considered a logic combining $\Box$ and $\Diamond$.
However, their obtained logic is stronger than \dik, as they considered a restricted classes of frames.
Moreover, it was non-appropriate from an intuitionistic point of view as $\Diamond$ becomes definable in terms of $\Box$. 
In other respect, Bo{\v{z}}i{\'c} and Do{\v{s}}en did not tackle the decidability issue. 
A logic related to \dik\ has been considered in~\cite{d1997grafting} in the context of substructural logics.
More recently the S4-extension of \dik\ has been shown to be decidable in~\cite{Balbiani:et:al:2021}.
In this paper, we consider \dik\ and some of its extensions with axioms characterizing~---~in models $(W,{\leq},{R},V)$~---~the seriality, the reflexivity and the transitivity of the accessibility relation $R$.
We provide complete axiomatizations for them with respect to appropriate classes of models.
The basic logic \dik\ is stronger than Wijesekera's \textbf{CCDL}, stronger than the \iml\ \fik\ which only assumes forward confluence on models~\cite{FIK-csl2024}.
It is also incomparable with \ik. 
It is noteworthy that \dik\ fails to satisfy the disjunction property.
However, unexpectedly, its extensions with axioms characterizing seriality or reflexivity of accessibility relations possess this property.
Turning to proof theory, we propose bi-nested sequent calculi for \dik\ and its extensions.  Nested sequent calculi for other \iml{s} are known since \cite{Galmiche:Salhi:2010,2013cut,marin2014label}. The bi-nested calculi use two kinds of nestings: the first one for representing $\geq$-upper worlds~---~as in~\cite{Fitting:2014}~---~and the second one for representing $R$-related worlds.
A sequent calculus with the same kind of nesting to capture an extension of \textbf{CCDL}\ has been presented in~\cite{anupam:2023} whereas a calculus for \ik\ with the same nesting was also preliminarily considered in~\cite{Marin:Maroles:2020}.
A bi-nested calculus with the same structure is proposed for the  logic \fik\ in \cite{FIK-csl2024} where the frame condition of forward confluence is captured by a suitable ``interaction'' rule. 
A calculus for \dik\ can be obtained from the calculus for \fik\ by adopting a ``local'' $\Box$, or by adding another ``interaction'' rule capturing the downward confluence condition.
We prove that the calculi provide a decision procedure for the logic \dik\ and some of its extensions.
Moreover, we show the semantic completeness of these calculi: from a single failed derivation under a suitable strategy, it is possible to extract a \emph{finite} countermodel of the given sequent. 
In addition, for the extensions of \dik\ with (\axiomd) or (\axiomt), we give a syntactic proof~---~via the calculi~---~of the disjunction property. 
These results demonstrate that bi-nested sequent calculus is a powerful and flexible tool.
It constitutes an alternative to labelled sequent calculus capable to treat uniformly various \iml{s}.

%% file: 01-logic.tex
\section{The logic}\label{sec:logic}
%
%
Let $\At$ be a set (with members called {\em atoms}\/ and denoted $p$, $q$, etc).
\begin{definition}[Formulas]
Let $\Fo$ be the set (with members called {\em formulas}\/ and denoted $A$, $B$, etc) of finite words over $\At{\cup}\{{\supset},{\top},{\bot},{\vee},{\wedge},{\square},{\lozenge},(,)\}$ defined by
$$A\ {::=}\ p{\mid}(A{\supset}A){\mid}{\top}{\mid}{\bot}{\mid}(A{\vee}A){\mid}(A{\wedge}A){\mid}{\square}A{\mid}{\lozenge}A$$
where $p$ ranges over $\At$.
For all $A\in\Fo$, when we write $\neg A$ we mean $A\supset\bot$.
\end{definition}
%
%
%
%
%
%
%
%
%
%
%
%
For all sets $\Gamma$ of formulas, let $\square\Gamma{=}\{A{\in}\Fo:\ \square A{\in}\Gamma\}$ and $\lozenge\Gamma{=}\{\lozenge A{\in}\Fo:\ A{\in}\Gamma\}$.
\begin{definition}[Frames]
A {\em frame}\/ is a relational structure $(W,{\leq},{R})$ where $W$ is a nonempty set of {\em worlds,} $\leq$ is a preorder on $W$ and ${R}$ is a binary relation on $W$.
%
%
%
%
%
%
A frame $(W,{\leq},{R})$ is {\em forward (resp. downward) confluent}\/ if ${\geq}{\circ}{R}{\subseteq}{R}{\circ}{\geq}$ (resp. ${\leq}{\circ}{R}{\subseteq}{R}{\circ}{\leq}$).
Let ${\mathcal C}_{\fcfra}$ be the class of forward confluent frames.
For all $X{\subseteq}\{\AxiomSer,\AxiomRef,\AxiomTra\}$, an $X$-frame is a frame $(W,\leq,R)$ such that $R$ is serial if $\AxiomSer{\in}X$, $R$ is reflexive if $\AxiomRef{\in}X$ and $R$ is transitive if $\AxiomTra{\in}X$.
Let ${\mathcal C}_{\fdcfra}^{X}$ be the class of forward and downward confluent $X$-frames.
We write ``${\mathcal C}_{\fdcfra}$'' instead of ``${\mathcal C}_{\fdcfra}^{\emptyset}$''.
%
%
\end{definition}
%
%
%
%
\begin{definition}[Valuations, models and truth conditions]
%
%
A {\em valuation on $(W,{\leq},{R})$}\/ is a function $V: \At\longrightarrow\wp(W)$ such that for all $p{\in}\At$, $V(p)$ is $\leq$-closed.
%
%
%
%
%
%
A {\em model based on $(W,{\leq},{R})$}\/ is a model of the form $(W,{\leq},{R},V)$.
%
%
%
%
%
%
In a model ${\mathcal M}{=}(W,{\leq},{R},V)$, for all $x\in W$ and for all $A\in \Fo$, the {\em satisfiability of $A$ at $x$ in ${\mathcal M}$}\/ (in symbols ${\mathcal M},x\Vdash A$) is defined as usual when $A$'s main connective is either $\top$, $\bot$, $\vee$ or $\wedge$ and as follows otherwise:
\begin{itemize}
\item ${\mathcal M},x\Vdash p$ if and only if $x\in V(p)$,
\item ${\mathcal M},x\Vdash A \supset B$ if and only if for all $x'\in W$ with $x\leq x'$, if ${\mathcal M},x'\Vdash A$ then ${\mathcal M},x'\Vdash B$,
\item ${\mathcal M},x\Vdash \square A$ if and only if for all $y\in W$, if $Rxy$, then ${\mathcal M},y\Vdash A$,
\item ${\mathcal M},x\Vdash \lozenge A$ if and only if there exists $y\in W$ such that $Rxy$ and ${\mathcal M},y\Vdash A$.
\end{itemize}
%
%
When ${\mathcal M}$ is clear from the context, we write ``$x \Vdash A$'' instead of ``${\mathcal M},x \Vdash A$''.
We define ``truth'' and ``validity'' as usual.
\end{definition}
\begin{lemma}[Heredity Property]\label{lemma:HB:monotonicity}
  Let $(W,{\leq},{R},V)$ be a forward and downward confluent model. For all $A \in \Fo$ and for all $x,x'\in W$, if $x\Vdash A$ and $x\leq x'$ then $x'\Vdash A$.
\end{lemma}
Our definition of $\Vdash$ differs from the definitions proposed by Fischer-Servi~\cite{Fischer:Servi:1984} and Wijesekera~\cite{Wijesekera:1990}: both in~\cite{Fischer:Servi:1984} and~\cite{Wijesekera:1990}, 
$x\Vdash \square A$ if and only if for all $x'\in W$ with $x\leq x'$ and for all $y\in W$ with $Rx'y$, it holds $y \Vdash A$ whereas in~\cite{Wijesekera:1990}, $x\Vdash \lozenge A$ if and only if for all $x'\in W$ with $x\leq x'$ then there exists $y\in W$ such that $Rx'y$ and $y\Vdash A$.
%
%
%
%
%
%
%
%
%
%
%
%
%
%
%
%
%
%
%
%
However,
\begin{proposition}\label{proposition:comparing:FS:and:W}
In ${\mathcal C}_{\fdcfra}$, our definition of $\Vdash$ determines the same satisfiability relation as the one determined by the definitions proposed in~\cite{Fischer:Servi:1984} and~\cite{Wijesekera:1990}.
\end{proposition}
From now on in this section, when we write frame (resp. model), we mean forward and downward confluent frame (resp. model).
\\
\\
Obviously, validity in ${\mathcal C}_{\fdcfra}$ is closed with respect to the following inference rules:
\[
  \AxiomC{$p\supset q, p$}
  \RightLabel{(\rm \textbf{MP})}
  \UnaryInfC{$q$}
  \DisplayProof
  \quad
  \AxiomC{$p$}
  \RightLabel{(\rm \textbf{NEC})}
  \UnaryInfC{$\square p$}
  \DisplayProof
\]
Moreover, the following formulas are valid in ${\mathcal C}_{\fdcfra}$:
$$
\begin{array}{ll}
  (\K_{\square}) \ \square(p\supset q)\supset(\square p\supset\square q)\qquad  &  (\K_{\Diamond}) \ \square(p\supset q)\supset(\Diamond p\supset\Diamond q) \\
  (\DP) \ \Diamond(p\vee q)\supset\Diamond p\vee\Diamond q & (\CD) \ \square(p\vee q)\supset\Diamond p\vee\square q\\
  (\NotPosBot) \ \neg\Diamond\bot &
\end{array}
$$
In ${\mathcal C}_{\fdcfra}^{\AxiomSer}$ (resp. ${\mathcal C}_{\fdcfra}^{\AxiomRef}$, ${\mathcal C}_{\fdcfra}^{\AxiomTra}$), the formula $\AxiomSer$ (resp. $\AxiomRef$, $\AxiomTra$) is valid:
%
$$
\begin{array}{ccc}
  (\AxiomSer)\ \lozenge\top \quad &
  (\AxiomRef)\ (\square p\supset p)\wedge(p\supset\lozenge p) \quad &
  (\AxiomTra)\ (\square p\supset\square\square p)\wedge(\lozenge\lozenge p\supset\lozenge p)
\end{array}
$$
%
%
Notice that axiom $\CD$ has also been considered in~\cite{Balbiani:et:al:2021} where it was called $\mathbf{CD}$ (for ``constant domain'') because of its relationship with the first-order formula $\forall x.(P(x){\vee}Q(x)){\supset}\exists x.P(x){\vee}\forall x.Q(x)$ which is intuitionistically valid when models with constant domains are considered.
\begin{definition}[Axiom system]
%
%
For all $X{\subseteq}\{\AxiomSer,\AxiomRef,\AxiomTra\}$, let \dik$X$ be the axiomatic system consisting of all standard axioms of \textbf{IPL}, the inference rules $\MP$ and $\NEC$, the axioms $\K_{\square}$, $\K_{\Diamond}$, $\NotPosBot$, $\DP$ and $\CD$ and containing in addition the axioms from $X$.
We write ``\dik'' instead of ``\dik$\emptyset$''.
Derivations are defined as usual.
We write $\vdash_{\mathbf{LIK}X}A$ when $A$ is \dik$X$-derivable.
The sets of all \dik$X$-derivable formulas will also be denoted \dik$X$.
\end{definition}
From now on in this section, let $X{\subseteq}\{\AxiomSer,\AxiomRef,\AxiomTra\}$.
\begin{lemma}\label{lemma:about:D:and:axioms}
If $\AxiomSer{\in}X$ or $\AxiomRef{\in}X$ then $\square p{\supset}\lozenge p$ and $\neg\square\bot$ are in \dik$X$.
\end{lemma}
%
%
\begin{theorem}\label{theorem:soundness:dfik}
\dik$X$-derivable formulas are ${\mathcal C}_{\fdcfra}^{X}$-validities.
%
%
\end{theorem}
%
%
%
%
%
%
%
%
%
%
%
%
%
%
%
%
%
%
We now prove the converse inclusion (Completeness) saying that every formula valid in ${\mathcal C}_{\fdcfra}^{X}$ is \dik$X$-derivable.
At the heart of our  completeness proof, there is  the concept of theory.
Let $\mathbf{L}{=}$\dik$X$.
\begin{definition}[Theories]
A \textit{theory} is a set of formulas containing $\mathbf{L}$ and closed with respect to \textbf{MP}.
A theory $\Gamma$ is \textit{proper} if $\bot\not\in\Gamma$.
A proper theory $\Gamma$ is \textit{prime} if for all formulas $A,B$, if $A\vee B\in\Gamma$ then either $A\in\Gamma$, or $B\in\Gamma$.
%
%
\end{definition}
%
%
%
%
%
%
%
%
%
%
%
%
%
%
%
%
%
%
\begin{lemma}\label{lemma:another:property:of:LIKD}
If $\AxiomSer{\in}X$ or $\AxiomRef{\in}X$ then for all theories $\Gamma$, $\lozenge\square\Gamma{\subseteq}\Gamma$.
\end{lemma}
%
%
%
%
%
%
%
%
%
%
%
%
%
%
%
%
%
%
%
%
%
%
%
%
%
%
%
%
%
%
%
%
%
%
%
%
%
%
%
%
%
%
%
%
%
%
\begin{definition}[Canonical model]
The {\em canonical model}\/ is the model $(W_{\mathbf{L}},$
$\leq_{\mathbf{L}},R_{\mathbf{L}},V_{\mathbf{L}})$ where
\begin{itemize}
\item 
$W_{\mathbf{L}}$ is the nonempty set of all prime theories,
\item
$\leq_{\mathbf{L}}$ is the partial order on $W_{\mathbf{L}}$ such that $\Gamma{\leq_{\mathbf{L}}}\Delta$ if and only if $\Gamma{\subseteq}\Delta$,
%
%
\item
$R_{\mathbf{L}}$ is the binary relation on $W_{\mathbf{L}}$ such that $R_{\mathbf{L}}\Gamma\Delta$ iff $\square\Gamma{\subseteq}\Delta$ and $\lozenge\Delta{\subseteq}\Gamma$,
\item $V_{\mathbf{L}}$ is the valuation on $W_{\mathbf{L}}$ such that for all $p{\in}\At$, $V_{\mathbf{L}}(p)=\{\Gamma{\in}W_{\mathbf{L}}:\ p{\in}\Gamma\}$.
\end{itemize}
%
%
%
%
%
%
\end{definition}
\begin{lemma}\label{fc:canonical:frame:is:forward:confluent}
\begin{enumerate}
\item $(W_{\mathbf{L}},\leq_{\mathbf{L}},R_{\mathbf{L}},V_{\mathbf{L}})$ is forward confluent,
\item $(W_{\mathbf{L}},\leq_{\mathbf{L}},R_{\mathbf{L}},V_{\mathbf{L}})$ is downward confluent,
\item if $\AxiomSer{\in}X$ (resp. $\AxiomRef{\in}X$, $\AxiomTra{\in}X$) then $(W_{\mathbf{L}},\leq_{\mathbf{L}},R_{\mathbf{L}},V_{\mathbf{L}})$ is serial (resp. reflexive, transitive).
\end{enumerate}
\end{lemma}
%
%
%
%
%
%
%
%
%
%
The proof of the completeness will be based on the following lemmas.
\begin{lemma}[Existence Lemma]\label{lemma:prime:proper:for:implication}
%
%
Let $\Gamma$ be a prime theory.
\begin{enumerate}
\item If $B\supset C\not\in\Gamma$ then there exists a prime theory $\Delta$ such that $\Gamma\subseteq\Delta$, $B\in\Delta$ and $C\not\in\Delta$,
\item if $\square B\not\in\Gamma$ then there exists a prime theory $\Delta$ such that $R_{\mathbf{L}}\Gamma\Delta$ and $B\not\in\Delta$,
\item if $\Diamond B\in\Gamma$ then there exists a prime theory $\Delta$ such that $R_{\mathbf{L}}\Gamma\Delta$ and $B\in\Delta$.
\end{enumerate}
\end{lemma}
%
%
%
%
%
%
%
%
%
%
%
%
%
%
\begin{lemma}[Truth Lemma]\label{lemma:truth:lemma}
For all formulas $A$ and for all $\Gamma\in W_{\mathbf{L}}$, $A\in\Gamma$ if and only if $(W_{\mathbf{L}},\leq_{\mathbf{L}},R_{\mathbf{L}},V_{\mathbf{L}}),\Gamma\Vdash A$.
\end{lemma}
%
%
%
%
From Lemma~\ref{lemma:truth:lemma}, we conclude.
\begin{theorem}[Completeness]\label{theorem:completeness:dfik}
All ${\mathcal C}_{\fdcfra}^{X}$-validities are \dik$X$-derivable.
\end{theorem}
%
%
Now, we show that \dik$X$ possesses the formal features that might be expected of an intuitionistic modal logic~\cite[Chapter~$3$]{Simpson:1994}.
%
%
%
%
%
%
%
%
%
%
%
%
%
%
%
%
\begin{proposition}\label{lemma:l:min:C1}
\begin{enumerate}
\item \dik$X$ is conservative over $\IPL$,
\item \dik$X$~contains all substitution instances of $\IPL$ and is closed with respect to modus ponens,
\item \dik$X$~has the disjunction property if and only if $\AxiomSer{\in}X$ or $\AxiomRef{\in}X$,
\item the addition of the law of excluded middle to \dik$X$~yields modal logic $\K$,
\item $\square$ and $\lozenge$ are independent in \dik$X$.
\end{enumerate}
\end{proposition}

%% file: 02-new-calculi.tex
\section{Bi-nested sequent calculi}\label{sec:calculi}
 
In this section, we present bi-nested calculi \cdikk~for \dik \ and its extensions for \did~ and \dit.
The calculi make use of two kinds of nesting representing  $\leq$-upper worlds and $R$-successors in the semantics, similarly to the calculus for \fik~presented in \cite{FIK-csl2024}.
The calculi contain two rules encoding forward and downward confluence.
However, we will show that the latter rule called $(\text{inter}_\downarrow)$ is admissible in \cdikk, so that by dropping this rule we still have a complete calculus for \dik.
But as we will see,  $(\text{inter}_\downarrow)$ rule~is  needed to prove the semantic completeness of the calculus and obtain counter-model extraction.
We also prove the disjunction property for the calculi for \did~and \dit.

In order to define the calculi we need some preliminary notions. 

\begin{definition}[Bi-nested sequent]\label{def-nested-sequent}
  A bi-nested sequent $S$ is defined as:\\
   - \ $\Rightarrow$ is a bi-nested sequent (the empty sequent);\\
   - \ $\Gamma\Rightarrow B_1,\ldots, B_k,[S_1],\ldots,[S_m],\langle T_1\rangle,\ldots,\langle T_n\rangle$ is a bi-nested sequent if all the 
    $B_1,\ldots, B_k$ are formulas, all the $S_1,\ldots,S_m$, $T_1,\ldots, T_n$ are bi-nested sequents where $k,m,n\geq 0$, 
    and $\Gamma$ is a finite (possibly empty) multi-set of formulas.
\end{definition}
We use $S$ and $T$ to denote a  bi-nested sequent and we call it simply a ``sequent" in the rest of this paper. The antecedent and succedent of a sequent $S$ are denoted by $Ant(S)$ and $Suc(S)$ respectively.

The notion of modal degree can be extended from a formula to a sequent. 

\begin{definition}[Modal degree]
  Modal degree for a formula $F$, denoted as $\textit{md}(F)$, is defined as usual. Further, let $\Gamma$ be a finite set of formulas, define $\textit{md}(\Gamma)=\textit{md}(\bigwedge \Gamma)$. 
  For a sequent $S=\Gamma\Rightarrow\Delta,[S_1],\ldots,[S_m],\langle T_1\rangle,\ldots,\langle T_n\rangle$, $\textit{md}(S)=\max\{\textit{md}(\Gamma),\textit{md}(\Delta),\textit{md}(S_1)+1,\ldots,\textit{md}(S_m)+1,\textit{md}(T_1),\ldots,\textit{md}(T_n)\}$.
\end{definition}

\begin{definition}[Context]\label{def-context}
  A context $G\{\}$ is inductively defined as follows:\\
  -  $\{\}$ is a context (the empty context). \\
- if $\Gamma\Rightarrow\Delta$ is a sequent and  $G'\{\}$ is a context then $\Gamma\Rightarrow\Delta, \langle G'\{\}\rangle$ is a context.\\
- if $\Gamma\Rightarrow\Delta$ is a sequent and  $G'\{\}$ is a context then $\Gamma\Rightarrow\Delta, [G'\{\}]$ is a context.
\end{definition}

\begin{example}\label{ex:context}
  Given a context $G\{\}=p\wedge q, \square r\Rightarrow \Diamond p,\langle \square p\Rightarrow [\Rightarrow q]\rangle,[\{\}]$ and a sequent $S=p\Rightarrow q\vee r,[r\Rightarrow s]$, we have $G\{S\}=p\wedge q, \square r\Rightarrow \Diamond p,\langle \square p\Rightarrow [\Rightarrow q]\rangle,[p\Rightarrow q\vee r,[r\Rightarrow s]]$. 
\end{example}

\begin{definition}[$\in^{\langle\cdot\rangle}, \in^{[\cdot]}, \in^+$-relation]
  Let $\Gamma_1\Rightarrow\Delta_1,\Gamma_2\Rightarrow\Delta_2$ be two sequents. We denote $\Gamma_1\Rightarrow\Delta_1 \in^{\langle\cdot\rangle}_0\Gamma_2\Rightarrow\Delta_2$ if $\langle \Gamma_1\Rightarrow\Delta_1 \rangle\in\Delta_2$ and let $\in^{\langle\cdot\rangle}$ be the transitive closure of $\in^{\langle\cdot\rangle}_0$. Relations $\in^{[\cdot]}_0$ and $\in^{[\cdot]}$ for modal blocks are  defined similarly. 
  Besides, let $\in^+_0 =~\in^{\langle\cdot\rangle}_0 \cup \in^{[\cdot]}_0$ and finally let $\in^+$ be the reflexive-transitive closure of $\in^+_0$. 
\end{definition}
Observe that when we say $S'\in^+ S$, it is equivalent to say that for some context $G$, $S = G\{S'\}$.

Some rules of the calculus propagate formulas in the antecedent (``positive part'') or the consequent (``negative part'') of sequents in a modal block. The two operators in the next definition single out the formulas that will be propagated. 

\begin{definition}[$\flat$-operator and $\#$-operator]\label{def:star-operator}
  Let $\Lambda\Rightarrow\Theta$ be a sequent and $Fm(\Theta)$ the multiset of formulas directly belonging to $\Theta$. 
  Let $\Theta^\flat=\emptyset$ if $\Theta$ is $[\cdot]$-free; 
  $\Theta^\flat= [\Phi_1\Rightarrow\Psi_1^\flat], \ldots, [\Phi_k\Rightarrow\Psi_k^\flat]$, if $\Theta=\Theta_0,[\Phi_1\Rightarrow\Psi_1], \ldots, [\Phi_k\Rightarrow\Psi_k]$ and $\Theta_0$ is $[\cdot]$-free.
  
  Dually let $\Rightarrow\Theta^\#= \ \Rightarrow Fm(\Theta)$ if $\Theta$ is $[\cdot]$-free; 
  $\Rightarrow\Theta^\#= \ \Rightarrow Fm(\Theta_0), [\Rightarrow\Psi_1^\#], \ldots, [\Rightarrow\Psi_k^\#]$ 
  if $\Theta=\Theta_0,[\Phi_1\Rightarrow\Psi_1], \ldots, [\Phi_k\Rightarrow\Psi_k]$ and $\Theta_0$ is $[\cdot]$-free.
\end{definition}


\begin{example}
  Consider the  sequent $G\{S\}=p\wedge q, \square r\Rightarrow \Diamond p,\langle \square p\Rightarrow [\Rightarrow q]\rangle,[p\Rightarrow q\vee r,[r\Rightarrow s]]$ of Example \ref{ex:context}, denote $Ant(G\{S\})$ and $Suc(G\{S\})$ by $\Lambda$ and $\Theta$ respectively, we can see by definition, 
  $\Lambda\Rightarrow\Theta^\flat=p\wedge q, \square r\Rightarrow [p\Rightarrow [r\Rightarrow ]]$ while 
  $\Rightarrow\Theta^\#= \ \Rightarrow\Diamond p,[\Rightarrow q\vee r,[\Rightarrow s]]$.
\end{example}


\begin{definition}
  Rules for the basic logic \dik~and its modal extensions are given in Figure \ref{cdik}, which consists of the basic calculus \cdikk~and modal rules corresponding to axioms (\axiomd), ($\axiomt_\Diamond$) and ($\axiomt_\square$). 
  We define \cdid~=~\cdikk~+~(\axiomd) and \cdit~=~\cdikk~+~($\axiomt_\square$)~+~($\axiomt_\Diamond$).
\end{definition}

\begin{figure}[!htb]
    \centering
    \noindent
    {\noindent \tiny \bf The basic calculus \cdikk:}
    \vspace{0.1cm}
    
    {\tiny
  $$
  \begin{array}{ccc}
      \AxiomC{}
      \RightLabel{\rm ($\bot_L$)}
      \UnaryInfC{$G\{\Gamma,\bot\Rightarrow\Delta\}$}
      \DisplayProof
      \qquad\qquad
      &
      \AxiomC{}
      \RightLabel{\rm ($\top_R$)}
      \UnaryInfC{$G\{\Gamma\Rightarrow\top,\Delta\}$}
      \DisplayProof
      \qquad
      &
      \AxiomC{}
      \RightLabel{\rm ($\text{id}$)}
      \UnaryInfC{$G\{\Gamma,p\Rightarrow\Delta,p\}$}
      \DisplayProof
  \end{array}
  $$
  \vspace{-.4cm}
  $$
  \begin{array}{cc}
    \AxiomC{$G\{A,B,\Gamma\Rightarrow\Delta\}$}
    \RightLabel{\rm ($\wedge_L$)}
    \UnaryInfC{$G\{A\wedge B,\Gamma\Rightarrow\Delta\}$}
    \DisplayProof
    \vspace{0.2cm}
    &
    \AxiomC{$G\{\Gamma\Rightarrow\Delta,A\}$}
    \AxiomC{$G\{\Gamma\Rightarrow\Delta,B\}$}
    \RightLabel{\rm ($\wedge_R$)}
    \BinaryInfC{$G\{\Gamma\Rightarrow\Delta,A\wedge B\}$}
    \DisplayProof
    \\
    \AxiomC{$G\{\Gamma,A\Rightarrow\Delta\}$}
    \AxiomC{$G\{\Gamma,B\Rightarrow\Delta\}$}
    \RightLabel{\rm ($\vee_L$)}
    \BinaryInfC{$G\{\Gamma,A\vee B\Rightarrow\Delta\}$}
    \DisplayProof
    \vspace{0.2cm}
    &
    \AxiomC{$G\{\Gamma\Rightarrow \Delta,A,B\}$}
    \RightLabel{\rm ($\vee_R$)}
    \UnaryInfC{$G\{\Gamma\Rightarrow \Delta,A\vee B\}$}
    \DisplayProof
    \\
    \AxiomC{$G\{\Gamma,A\supset B\Rightarrow A,\Delta\}$}
    \AxiomC{$G\{\Gamma,B\Rightarrow\Delta\}$}
    \RightLabel{\rm ($\supset_L$)}
    \BinaryInfC{$G\{\Gamma,A\supset B\Rightarrow\Delta\}$}
    \DisplayProof
    \vspace{0.2cm}
    &
    \AxiomC{$G\{\Gamma\Rightarrow \Delta, \langle A\Rightarrow B\rangle\}$}
    \RightLabel{\rm ($\supset_R$)}
    \UnaryInfC{$G\{\Gamma\Rightarrow \Delta, A\supset B\}$}
    \DisplayProof
    \\
    \AxiomC{$G\{\Gamma,\square A\Rightarrow\Delta,[\Sigma,A\Rightarrow \Pi]\}$}
    \RightLabel{\rm ($\square_L$)}
    \UnaryInfC{$G\{\Gamma,\square A\Rightarrow\Delta,[\Sigma\Rightarrow \Pi]\}$}
    \DisplayProof
    \vspace{0.2cm}
    &
    \AxiomC{$G\{\Gamma\Rightarrow\Delta,[\Rightarrow A]\}$}
    \RightLabel{($\square_{R}$)}
    \UnaryInfC{$G\{\Gamma\Rightarrow\Delta,\square A\}$}
    \DisplayProof
    \\  
    \AxiomC{$G\{\Gamma\Rightarrow\Delta,[A\Rightarrow]\}$}
    \RightLabel{\rm ($\Diamond_L$)}
    \UnaryInfC{$G\{\Gamma,\Diamond A\Rightarrow\Delta\}$}
    \DisplayProof
    &
    \AxiomC{$G\{\Gamma\Rightarrow\Delta,\Diamond A,[\Sigma\Rightarrow\Pi,A]\}$}
    \RightLabel{\rm ($\Diamond_R$)}
    \UnaryInfC{$G\{\Gamma\Rightarrow\Delta,\Diamond A,[\Sigma\Rightarrow\Pi]\}$}
    \DisplayProof
  \end{array}
  $$
  \vspace{-.1cm}
  $$
  \begin{array}{c}
    \AxiomC{$G\{\Gamma,\Gamma'\Rightarrow\Delta,\langle\Gamma',\Sigma\Rightarrow\Pi\rangle\}$}
    \RightLabel{\rm (\text{trans})}
    \UnaryInfC{$G\{\Gamma,\Gamma'\Rightarrow\Delta,\langle\Sigma\Rightarrow\Pi\rangle\}$}
    \DisplayProof
    \vspace{0.2cm}
    \\
    \AxiomC{$G\{\Gamma\Rightarrow\Delta,\langle\Sigma\Rightarrow\Pi,[\Lambda\Rightarrow\Theta^\flat]\rangle,[\Lambda\Rightarrow\Theta]\}$}
    \RightLabel{$(\text{inter}_{\rightarrow})$}
    \UnaryInfC{$G\{\Gamma\Rightarrow\Delta,\langle\Sigma\Rightarrow\Pi\rangle,[\Lambda\Rightarrow\Theta]\}$}
    \DisplayProof
    \vspace{0.2cm}
    \\
    \AxiomC{$G\{\Gamma\Rightarrow\Delta,\langle\Sigma\Rightarrow\Pi,[\Lambda\Rightarrow\Theta]\rangle,[\Rightarrow\Theta^\#]\}$}
    \RightLabel{($\text{inter}_{\downarrow}$)}
    \UnaryInfC{$G\{\Gamma\Rightarrow\Delta,\langle\Sigma\Rightarrow\Pi, [\Lambda\Rightarrow\Theta]\rangle\}$}
    \DisplayProof
  \end{array}  
  $$
    }
  \vspace{0.1cm}
  
  {
  \tiny
  {\bf \noindent Modal rules for the extensions:}
  $$
    \begin{array}{ccc}
      \AxiomC{$G\{\Gamma\Rightarrow\Delta,[\Rightarrow]\}$}
      \RightLabel{(\axiomd)}
      \UnaryInfC{$G\{\Gamma\Rightarrow\Delta\}$}
      \DisplayProof
      \quad
      \AxiomC{$G\{\Gamma,\square A, A\Rightarrow\Delta\}$}
      \RightLabel{($\axiomt_\square$)}
      \UnaryInfC{$G\{\Gamma,\square A\Rightarrow\Delta\}$}
      \DisplayProof
      \quad
      \AxiomC{$G\{\Gamma\Rightarrow\Delta,\Diamond A, A\}$}
      \RightLabel{($\axiomt_\Diamond$)}
      \UnaryInfC{$G\{\Gamma\Rightarrow\Delta,\Diamond A\}$}
      \DisplayProof
    \end{array}
  $$
  }
  \caption{Bi-nested rules for local intuitionistic modal logics}\label{cdik}
\end{figure}

Here are some remarks on the rules. 
Reading the rule upwards, the rule ($\supset_R$) introduces an implication block $\langle \cdot \rangle$ while the rules ($\Diamond_L$) and ($\Box_R$) introduce a modal block $[\cdot]$. 
Observe that the ($\Box_R$) rule corresponds to the ``local" interpretation of $\Box$. The rule $(\text{inter}_{\rightarrow})$ is intended to capture  Forward Confluence, whereas the rule ($\text{inter}_{\downarrow}$)   Downward Confluence. Finally the (trans) rule captures the Hereditary Property. All the rules of \cdikk,  except ($\Box_R$) and ($\text{inter}_{\downarrow}$)  belong to the calculus  \cfik~for the logic \fik~\cite{FIK-csl2024}, we will discuss the relation between the two calculi later in the section. 
	
We can verify that each axiom  of   Section \ref{sec:logic} is provable in \cdikk. We  give the example of  axiom (\textbf{RV}) here. 

\begin{example}
  We show $\square(p\vee q)\Rightarrow\Diamond p\vee\square q$ is provable. 
    \[
    \tiny
      \AxiomC{}
      \RightLabel{(id)}
      \UnaryInfC{$\square(p\vee q)\Rightarrow\Diamond p,[p\Rightarrow q,p]$}
      \AxiomC{}
      \RightLabel{(id)}
      \UnaryInfC{$\square(p\vee q)\Rightarrow\Diamond p,[q\Rightarrow q,p]$}
      \RightLabel{($\vee_L$)}
      \BinaryInfC{$\square(p\vee q)\Rightarrow\Diamond p,[p\vee q\Rightarrow q,p]$}
      \RightLabel{($\square_L$)}
      \UnaryInfC{$\square(p\vee q)\Rightarrow\Diamond p,[\Rightarrow q,p]$}
      \RightLabel{($\Diamond_R$)}
      \UnaryInfC{$\square(p\vee q)\Rightarrow\Diamond p,[\Rightarrow q]$}
      \RightLabel{($\square_R$)}
      \UnaryInfC{$\square(p\vee q)\Rightarrow\Diamond p,\square q$}
      \RightLabel{($\vee_R$)}
      \UnaryInfC{$\square(p\vee q)\Rightarrow\Diamond p\vee\square q$}
      \DisplayProof
    \]
\end{example}

We now show that \cdikk~is sound with respect to the semantics.
First we extend the forcing relation $\Vdash$ to sequents and blocks therein. 

\begin{definition}
Let $\mathcal{M}=(W,\leq,R,V)$ be a model and $x\in W$. The relation $\Vdash$ is extended to sequents as follows:
\begin{quote}
$\mathcal{M}, x\not\Vdash \emptyset$\\
$\mathcal{M}, x \Vdash [T]$ if for every $y$ with $Rxy$, $\mathcal{M}, y \Vdash T$\\
$\mathcal{M}, x \Vdash \langle T \rangle$ if for every $x'$ with $x\leq x'$, $\mathcal{M}, x' \Vdash T$\\
$\mathcal{M}, x \Vdash \Gamma \Rightarrow \Delta$ if either $\mathcal{M}, x \not \Vdash A$ for some $A\in \Gamma$ or $\mathcal{M}, x \Vdash {\cal O}$ for some ${\cal O} \in \Delta$, where ${\cal O}$ is a formula or a block. 
\end{quote}
We say $S$ is {\rm valid} in $\+M$ iff $\forall w\in W$, we have $\+M,w\Vdash S$.
We say $S$ is {\rm valid} iff it is valid in every model.
\end{definition}


\begin{definition}
  For a rule ($r$) of the form $\frac{G\{S_1\} \quad G\{S_2\}}{G\{S\}}$ or $\frac{G\{S_1\}}{G\{S\}}$, 
  we say ($r$) is valid if $x\Vdash G\{S_i\}$ implies $x\Vdash G\{S\}$. 
\end{definition}

We obtain the soundness of \cdikk\ by verifying the validity of each rule.
The soundness of \cdid~and \cdit~can be proven similarly. 

\begin{theorem}[Soundness of \cdikk]\label{thm:soundness-cdik}
  If a sequent $S$ is provable in \cdikk, then it is a validity in \dik.
\end{theorem}

Next, we show that the rule ($\text{inter}_{\downarrow}$) is admissible in the calculus  \local = \cdikk$\backslash\{(\text{inter}_{\downarrow})\}$;  the proof can be extended to the modal extensions. 
In order to prove this, we need some preliminary facts. 
First, weakening and contraction rules $(w_L)(w_R)(c_L)(c_L)$ defined as usual are \emph{height-preserving} (hp) admissible in \local, not only applied to formulas but also to blocks. Moreover, extended weakening rules $\frac{S}{G\{S\}}$, $\frac{G\{\Gamma\Rightarrow\Delta^\flat\}}{G\{\Gamma\Rightarrow\Delta\}}$, $\frac{G\{\Gamma\Rightarrow\Delta^\#\}}{G\{\Gamma\Rightarrow\Delta\}}$ are hp-admissible as well.

\begin{proposition}\label{prop:lik-eq-fragment}
  The $(\text{inter}_{\downarrow})$ rule is admissible in \local. Consequently, a sequent $S$ is provable in  \cdikk \ if and only if  $S$ is provable in \local.
\end{proposition}

As mentioned above, all the rules in \cdikk, except ($\Box_R$) and ($\text{inter}_{\downarrow}$) belong to the calculus \cfik~for the logic \fik~\cite{FIK-csl2024}. As a difference with \dik, the logic \fik \ adopts the global forcing condition for $\Box$ as in \cite{Fischer:Servi:1984,Simpson:1994,Wijesekera:1990} 
and only forward confluence on the frame. The ``global'' ($\Box_R$) rule in \cfik \ is $\frac{G\{\Gamma\Rightarrow\Delta,\langle\Rightarrow[\Rightarrow A]\rangle\}}{G\{\Gamma\Rightarrow\Delta,\square A\}}$. 
It can be proved that this rule is admissible in \local and on the opposite direction, the ``local" rule for $(\Box_R)$ of  \cdikk \ is admissible in \cfik + ($\text{inter}_{\downarrow}$). Thus    an equivalent calculus for \dik \ can be obtained in\emph{ a modular way} from the one for \fik.

We end this section by considering the disjunction property. 
For simplicity, we only work in \local~and its extensions. Let \cdidd~=~\local~+~(\axiomd) and \cditt~=~\local~+~($\axiomt_\square$)~+~($\axiomt_\Diamond$). 
Consider the formula $\square \bot\vee \Diamond\top$ which is provable in \local, but it is easy to see that neither  $\square\bot$, nor  $\Diamond \top$, are provable\footnote{We thank Tiziano Dalmonte for having suggested this counterexample.}. However, this counterexample does not hold for in \did~and for \dit~since $\Diamond \top$ is provable in both calculi. 
We  show indeed that the disjunction property holds for both \cdidd \ and \cditt.
The key fact is expressed by the following lemma: 

\begin{lemma}\label{lemma:dp-cdid}
  Suppose $S =\ \Rightarrow A_1, \ldots, A_m, \langle G_1\rangle, \ldots, \langle G_n\rangle,[H_1],\ldots,[H_l]$ is provable in \cdidd (resp. \cditt), where $A_i$'s are formulas, $G_j$ and $H_k$'s are sequents. 
  Furthermore, we assume that each $H_k$ is of the form $\Rightarrow\Theta_k$ and for each sequent $T\in^{[\cdot]} H_k$, $T$ has an empty antecedent. Then either $\Rightarrow A_i$, or $\Rightarrow \langle G_j\rangle$, or $\Rightarrow[H_k]$  is provable in \cdidd (resp. \cditt) for some $i\leq m, j\leq n, k\leq l$ .
\end{lemma}

We obtain the disjunction property by an obvious application of the  lemma. 

\begin{proposition}[Disjunction property for \cdidd and \cditt]\label{prop:dp}
  For any formulas $A,B$, if $\Rightarrow A \lor B$ is provable in \cdidd (resp. \cditt), then either $\Rightarrow A$ or $\Rightarrow  B$ is provable \cdidd (resp. \cditt).
\end{proposition}

%% file: 03-terminationNEW.tex
\section{Termination}\label{sec:termination}
In this section we define decision procedures for \dik~as well as its extensions \did~and \dit~based on the previous calculi. We treat first \dik, then at the end of the section we briefly describe how to adopt the the procedure to the extensions. 
The terminating proof-search procedure is needed to prove the semantic completeness of the calculus and we will show how to build a (finite) countermodel of the sequent at the root of a derivation, whenever a derivation fails. 

We have introduced two calculi for \dik, namely \cdikk \ and \local. For \local, we can obtain a terminating proof-procedure by adapting the one in \cite{FIK-csl2024} for the calculus of \fik. Actually, the decision procedure for \local is remarkably simpler than the one for \fik, as ``blocking'' is not needed to prevent loops. For \cdikk \ we need some extra work. Although the calculi \cdikk~and \local~are equivalent for provability, \local~does not allow a countermodel extraction so we have to consider the latter for proving semantic completeness.  

\remove{
\begin{example}
	Let $A = \Box a \supset (\Box b \supset \Box c)$ where $a,b,c$ are atoms. Then in \local we generate the following saturated sequent $S_0$ \footnote{To simplify we do not use here the cumulative calculus in Definition \ref{cumulcalc}. }:
		$\Box a \Rightarrow  \langle \Box a, \Box b \Rightarrow \Box c, [a, b \Rightarrow c]\rangle$. 
We will identify ``worlds" with sequent and blocks therein, in this case in addition to $S_0$, we further have $S_1 = \Box a, \Box b \Rightarrow \Box c, [a, b \Rightarrow c]$ and $S_2 = a,b \Rightarrow c$. And let $S_0 \leq S_1$ and $R S_1 S_2$. However this model not satisfy the downward confluence, and it is not clear how to repair it. On the other hand, working in \cdikk, we get the saturated sequent:
	$\Box a \Rightarrow \langle \Box a, \Box b \Rightarrow \Box c, [a, b \Rightarrow c]\rangle, [a \Rightarrow c]$
Let additionally $S_3 = a \Rightarrow c$;  the "world" $S_3$ is created by $(\text{inter}_\downarrow)$ and further propagation. We also have $S_3\leq S_2$ and $RS_0 S_3$ so that the model satisfies both downward and forward confluence.
\end{example}
}

Our ultimate aim is to build a countermodel from a failed derivation, the main ingredient is the pre-order relation $\leq$ in the model construction. This relation is specified by the following notion of \emph{structural inclusion} between sequents, which is also used in defining the saturation conditions needed for termination. 

\begin{definition}[Structural inclusion $\subseteq^\-S$]\label{def:saturation-inclusion}
	Let $S_1=\Gamma_1\Rightarrow\Delta_1,S_2=\Gamma_2\Rightarrow\Delta_2$ be two sequents. We say that $S_1$ is \textit{structurally included} in $S_2$, denoted by  $S_1\subseteq^{\-S} S_2$, whenever:  
	\vspace{-0.2cm}
	
	\begin{enumerate}
		\item[(i).] $\Gamma_1\subseteq\Gamma_2$; 
		\item[(ii).] for each $[\Lambda_1\Rightarrow\Theta_1]\in \Delta_1$, there exists $[\Lambda_2\Rightarrow\Theta_2]\in\Delta_2$ such that $\Lambda_1\Rightarrow\Theta_1\subseteq^{\-S} \Lambda_2\Rightarrow\Theta_2$;
		\item[(iii).] for each $[\Lambda_2\Rightarrow\Theta_2]\in\Delta_2$, there exists $[\Lambda_1\Rightarrow\Theta_1]\in \Delta_1$ such that $\Lambda_1\Rightarrow\Theta_1\subseteq^{\-S} \Lambda_2\Rightarrow\Theta_2$. 
	\end{enumerate}
\end{definition}
	\vspace{-0.2cm}

It is easy to see $\subseteq^\-S$ is both reflexive and transitive. 

We now define an equivalent variant \textbf{C}\cdikk \ of \cdikk \ which adopts a cumulative version of the rules and some bookkeeping. Moreover the rule of implication $(\supset_{R})$ is modified to prevent loops. This is the calculus that will be used as a base for the decision procedure and then semantic completeness. First we reformulate the $\#$-operator in the following way, annotating the $\#$-sequents by the full sequent where it comes from. 

\begin{definition}
  Let $Fm(\Theta)$ be the multiset of formulas directly belonging to $\Theta$. We define the $\#$-operator with annotation as follows:
  	\vspace{-0.2cm}
  
  \begin{itemize}
    \item $\Rightarrow_{\Lambda\Rightarrow\Theta}\Theta^\#= \ \Rightarrow Fm(\Theta)$ if $\Theta$ is $[\cdot]$-free;
    \item $\Rightarrow_{\Lambda\Rightarrow\Theta}\Theta^\#= \ \Rightarrow Fm(\Theta_0), [\Rightarrow_{\Phi_1\Rightarrow\Psi_1}\Psi_1^\#], \ldots, [\Rightarrow_{\Phi_k\Rightarrow\Psi_k}\Psi_k^\#]$ 
    if $\Theta=\Theta_0,[\Phi_1\Rightarrow\Psi_1], \ldots, [\Phi_k\Rightarrow\Psi_k]$ and $\Theta_0$ is $[\cdot]$-free.
  \end{itemize}
\end{definition}

The $\#$-sequents are generated only by ($\text{inter}_{\downarrow}$), and we use the annotation "track" the implication block from which a $\#$-sequent is generated. We omit the annotation and  write simply $\Rightarrow\Theta^\#$ whenever we do not need to track an application of ($\text{inter}_{\downarrow}$). 

\begin{definition}[The $\#$-annotated cumulative calculus \textbf{C}\cdikk]\label{cumulcalc}
  The cumulative calculus {\rm \textbf{C}\cdikk}~acts on set-based sequents, where a set-based sequent $S = \Gamma \Rightarrow \Delta$ is defined as in definition \ref{def-nested-sequent}, but $\Gamma$ is a {\em set} of formulas and $\Delta$ is a {\em set} of formulas and/or blocks (containing set-based sequents).
  The rules are as follows: 
  \begin{itemize}
    \item $(\bot_L), \ (\top_R), \ (\text{id}),\ (\square_L),\ (\Diamond_R)$, (trans) and $(\text{inter}_{\rightarrow})$ as in  \cdikk. 
    \item  $(\supset_R)$ is replaced by  $(\supset_{R_1}')$ and $(\supset_{R_2}')$ for $A\in\Gamma$ or $A\notin\Gamma$ respectively:
    {
    \[
    \AxiomC{$G\{\Gamma\Rightarrow\Delta,A\supset B,B\}$}
    \RightLabel{ ($A\in \Gamma$)}
    \UnaryInfC{$G\{\Gamma\Rightarrow\Delta,A\supset B\}$}
    \DisplayProof
    \quad
    \AxiomC{$G\{\Gamma\Rightarrow \Delta,A\supset B,\langle A\Rightarrow B\rangle\}$}
    \RightLabel{($A\notin\Gamma$)}
    \UnaryInfC{$G\{\Gamma\Rightarrow\Delta,A\supset B\}$}
    \DisplayProof
    \]
    }
    \item ($\text{inter}_{\downarrow}$) is replaced by the following annotated rule:
    \[
    \AxiomC{$G\{\Gamma\Rightarrow\Delta,\langle\Sigma\Rightarrow\Pi,[\Lambda\Rightarrow\Theta]\rangle,[\Rightarrow_{\Lambda\Rightarrow\Theta}\Theta^\#]\}$}
    \RightLabel{($\text{inter}'_{\downarrow}$)}
    \UnaryInfC{$G\{\Gamma\Rightarrow\Delta,\langle\Sigma\Rightarrow\Pi, [\Lambda\Rightarrow\Theta]\rangle\}$}
    \DisplayProof
    \]
    \item The other rules in \cdikk~are modified by keeping the principal formula in the premises. For example, the cumulative versions of $(\wedge_L),\ (\square_R)$ are:
    \[
      \AxiomC{$G\{A,B,A\wedge B,\Gamma\Rightarrow\Delta\}$}
    \RightLabel{\rm ($\wedge_L'$)}
    \UnaryInfC{$G\{A\wedge B,\Gamma\Rightarrow\Delta\}$}
    \DisplayProof
    \quad
    \AxiomC{$G\{\Gamma\Rightarrow\Delta,\square A,[\Rightarrow A]\}$}
    \RightLabel{\rm ($\square_R'$)}
    \UnaryInfC{$G\{\Gamma\Rightarrow\Delta,\square A\}$}
    \DisplayProof
    \]
  \end{itemize}
\end{definition}

Given the admissibility of weakening and contraction in \cdikk, the following proposition is a direct consequence.

\begin{proposition}
  A sequent $S$ is provable in \cdikk~ iff $S$ is provable in {\rm \textbf{C}\cdikk}.
\end{proposition}

We introduce  saturation conditions for each rule in \textbf{C}\cdikk. They are needed for both termination and counter-model extraction. 

\begin{definition}[Saturation conditions]
  Let $S=\Gamma\Rightarrow\Delta$ be a sequent. We say $S$ satisfies the saturation condition on the top level with respect to 
  \begin{itemize}
   \item ($\supset_R$) If $A\supset B\in\Delta$, then either $A\in\Gamma$ and $B\in\Delta$, or there is $\langle\Sigma\Rightarrow\Pi\rangle\in\Delta$ with $A\in\Sigma$ and $B\in\Pi$.
    \item ($\Diamond_R$): If $\Diamond A\in \Delta$ and $[\Sigma\Rightarrow\Pi]\in\Delta$, then $A\in\Pi$.
    \item ($\Diamond_L$): If $\Diamond A\in\Gamma$, then there is $[\Sigma\Rightarrow\Pi]\in\Delta$ with $A\in\Sigma$.
    \item ($\square_{R}$): if $\square A\in \Delta$, then there is $[\Lambda\Rightarrow\Theta]\in\Delta$ with $A\in\Theta$.
    \item ($\square_L$): If $\square A\in\Gamma$ and $[\Sigma\Rightarrow\Pi]\in\Delta$, then $A\in \Sigma$.
    \item ($\text{inter}_{\downarrow}$): if $\langle\Sigma\Rightarrow\Pi,[\Lambda\Rightarrow\Theta]\rangle\in \Delta$, then there is $[\Phi\Rightarrow\Psi]\in\Delta$ s.t. $\Phi\Rightarrow\Psi\subseteq^{\-S}\Lambda\Rightarrow\Theta$.
    \item ($\text{inter}_{\rightarrow}$): if $\langle\Sigma\Rightarrow\Pi\rangle,[\Lambda\Rightarrow\Theta]\in\Delta$, then there is $[\Phi\Rightarrow\Psi]\in\Pi$ s.t. $\Lambda\Rightarrow\Theta\subseteq^{\-S}\Phi\Rightarrow\Psi$.
    \item ($\text{trans}$): if $\langle\Sigma\Rightarrow\Pi\rangle\in\Delta$, then $\Gamma\subseteq\Sigma$.
  \end{itemize}
  Saturation conditions for the other propositional rules are defined as usual.
\end{definition}

\begin{proposition}\label{sequent-inclusion}
  Let $S=\Gamma\Rightarrow\Delta$. If $S$ is saturated with respect to $(\text{trans})$, $(\text{inter}_{\rightarrow})$ and ($\text{inter}_{\downarrow}$), then for $\langle\Sigma\Rightarrow\Pi\rangle\in\Delta$, we have $\Gamma\Rightarrow\Delta\subseteq^\-S\Sigma\Rightarrow\Pi$.
\end{proposition}

We say that a backward application of a rule ($r$) to a sequent $S$ is {\em redundant} if $S$ already satisfies the corresponding saturation condition associated with ($r$). In order to define a terminating proof-search strategy based on \textbf{C}\cdikk, we first impose the following constraints: 

\indent (i)  {\em No rule is  applied to an axiom} and (ii) {\em No rule is applied redundantly.}

\noindent However there is a problem: Backward proof search only respecting these constraints does not necessarily ensure that \emph{any} leaf of a derivation, to which no rule can be applied non-redundantly, satisfies \emph{all} the saturation conditions for rules in \textbf{C}\cdikk. This is a significant difference from the calculus of \fik \  in \cite{FIK-csl2024}. 
The problematic case is the saturation condition for the ($\text{inter}_{\downarrow}$) rule. 

\begin{example}\label{ex:5}
	Let us consider the sequent $\square(p\vee q)\Rightarrow \square r\supset\square s$. After some preliminary steps, we get two sequents: \\
(i). $\square(p\vee q)\Rightarrow \square r\supset\square s,\langle\square(p\vee q),\square r\Rightarrow \square s,[p\vee q,p,r\Rightarrow s]\rangle$.\\
(ii). $\square(p\vee q)\Rightarrow \square r\supset\square s,\langle\square(p\vee q),\square r\Rightarrow \square s,[p\vee q,q,r\Rightarrow s]\rangle$.\\
  Suppose we select (i) and then apply $(\text{inter}_\downarrow)$  obtaining (i'): $\square(p\vee q)\Rightarrow \square r\supset\square s,\langle\square(p\vee q),\square r\Rightarrow \square s,[p\vee q,p,r\Rightarrow s]\rangle,[\Rightarrow s]$. 
	After applying $(\Box_L)$, $(\lor_L)$ and $(\text{inter}_{\rightarrow}$), we obtain: \\
(iii).$\square(p\vee q)\Rightarrow \square r\supset\square s,\langle\square(p\vee q),\square r\Rightarrow \square s,[p\vee q,p,r\Rightarrow s]\rangle,[p\vee q,p\Rightarrow s]$.\\
(iv) $\square(p\vee q)\Rightarrow \square r\supset\square s,\langle\square(p\vee q),\square r\Rightarrow \square s,[p\vee q,p,r\Rightarrow s], [p\vee q,q\Rightarrow]\rangle$, \\ $[p\vee q,q\Rightarrow s]$.\\
We can see that (iii) satisfies the saturation condition for ($\text{inter}_{\downarrow}$), as $p\vee q,p\Rightarrow s\subseteq^\-S p\vee q,p,r\Rightarrow s$ but  (iv) does not, since there is no $[\Phi\Rightarrow\Psi]$ s.t. $\Phi\Rightarrow\Psi\subseteq^{\-S} p\vee q,p,r\Rightarrow s$. Sequent (iv) would not give a model satisfying  Downward Confluence\footnote{Observe that a disallowed \emph{redundant} application of ($\text{inter}_{\downarrow}$)  to the block $[p\vee q,q\Rightarrow]$ would not help, as it would reproduce the branching. }. This example also shows the inadequacy of \local \ for semantic completeness, as sequent expansion  in \local terminates with (i) and (ii), which would not define a model satisfying Downward Confluence. 
\end{example}
This means that not all "branches" in a derivation lead to sequents that produce a correct model. In order to get  a "correct"  counter-model, we need a mechanism that selects the branch that ensures the saturation condition for ($\text{inter}_{\downarrow}$). This is provided by the  tracking mechanism and realization procedure defined below.
 
\begin{definition}[Tracking record based on $\in^{[\cdot]}$]
  Let $S$ be a set-based sequent which is saturated with respect to all the left rules in {\rm \textbf{C}}\cdikk. Take an arbitrary set of formulas, denoted as $\Gamma$. Let $\Omega=\{T~|~T=S~or~T\in^{[\cdot]}S\}$. 
  For each $T\in \Omega$, we define $\mathfrak{G}_S(T,\Gamma)$, the \textit{$\in^{[\cdot]}$-based tracking record} of $\Gamma$ in $S$, which is a subset of $Ant(T)$ as follows: 
  
  \vspace{-0.2cm}
  \begin{itemize}
    \item $\mathfrak{G}_S(S,\Gamma)=\Gamma\cap Ant(S)$;
    \item If $T\in^{[\cdot]}_0 T'$ for some $T'\in\Omega$, let $\mathfrak{G}_S(T,\Gamma)$ be the minimal set such that
    \begin{itemize}
      \item if $\square A\in \mathfrak{G}_S(T',\Gamma)$, then $A\in\mathfrak{G}_S(T,\Gamma)$;
      \item if $\Diamond A\in \mathfrak{G}_S(T',\Gamma)$ and $A\in Ant(T)$, then $A\in\mathfrak{G}_S(T,\Gamma)$;
      \item if $A\wedge B\in \mathfrak{G}_S(T,\Gamma)$, then $A,B\in \mathfrak{G}_S(T,\Gamma)$;
      \item if $A\vee B\in \mathfrak{G}_S(T,\Gamma)$ and $A\in Ant(T)$, then $A\in \mathfrak{G}_S(T,\Gamma)$;
      \item if $A\supset B\in \mathfrak{G}_S(T,\Gamma)$ and $B\in Ant(T)$, then $B\in \mathfrak{G}_S(T,\Gamma)$.
    \end{itemize}
  \end{itemize}  
\end{definition}

Tracking record is used to control rule applications to and within a block created by ($\text{inter}_{\downarrow}$), preserving the saturation condition associated to it. 

\begin{definition}[Realization]\label{def:realization}
  Let $S=\Gamma\Rightarrow\Delta,\langle S_1\rangle,[S_2]$, where $S_1=\Sigma\Rightarrow\Pi,[\Lambda\Rightarrow\Theta]$, $S_2= \ \Rightarrow_{\Lambda\Rightarrow\Theta}\Theta^\#$ and $\Gamma\subseteq\Sigma$. 
  Moreover, we assume that $S_1$ is saturated with respect to all the left rules in {\rm \textbf{C}\cdikk}. 
  Using  the $\in^{[\cdot]}$-based tracking record of $\Gamma$ in $S_1$, we define the \textit{realization} of the block $[S_2]$ in $S$ as follows: 
  \begin{enumerate}
    \item[(i).] First for each $T\in^+ S_2$, define the realization function $f_{S_1}(T)$.\\
    \noindent By definition, $T$ is of the form $\Rightarrow_{\Phi\Rightarrow\Psi}\Psi^\#$ for some $\Phi\Rightarrow\Psi\in^+\Lambda\Rightarrow\Theta$. $f_{S_1}(T)$ is defined inductively on the structure of $\Psi^\#$ as follows: 
    \begin{itemize}
      \item if $\Psi^\#$ is block-free, then $f_{S_1}(T)=\mathfrak{G}(\Phi\Rightarrow\Psi,\Gamma)\Rightarrow\Psi^\#$.
      \item otherwise $\Psi^\#= \Psi_0, [T_1], \ldots, [T_k]$ where $\Psi_0$ is a set of formulas, then 
      $f_{S_1}(T)=\mathfrak{G}(\Phi\Rightarrow\Psi,\Gamma)\Rightarrow \Psi_0, [f_{S_1}(T_1)], \ldots, [f_{S_1}(T_k)]$.
    \end{itemize}
    \item[(ii).] With $f_{S_1}(S_2)$, the realization of $[S_2]$ in $S$ is $\Gamma\Rightarrow\Delta,\langle S_1\rangle,[f_{S_1}(S_2)]$. 
  \end{enumerate}
\end{definition}

As the next proposition shows (iii), the expansion produced by a realization procedure is not an extra logical step: it  can be obtained by applying the rules of the calculus by choosing the right branch.

\begin{proposition}\label{prop:realization}
  Let $S=\Gamma\Rightarrow\Delta,\langle S_1\rangle,[S_2]$, where $S_1=\Sigma\Rightarrow\Pi,[\Lambda\Rightarrow\Theta]$ and $S_2= \ \Rightarrow_{\Lambda\Rightarrow\Theta}\Theta^\#$ and $\Gamma\subseteq\Sigma$. 
  If $S_1$ is saturated with respect to all the left rules in {\rm \textbf{C}}\cdikk, then for the sequent $S'=\Gamma\Rightarrow\Delta,\langle S_1\rangle,[f_{S_1}(S_2)]$ which is obtained by the realization procedure in Definition \ref{def:realization}, we have
  \begin{enumerate}
    \item[(i).] $S'$ is saturated with respect to all the left rules applied to or within $[f_{S_1}(S_2)]$;
    \item[(ii).] $f_{S_1}(S_2)\subseteq^\-S \Lambda\Rightarrow\Theta$;
    \item[(iii).] $S'$ can be obtained by applying left rules of  {\rm \textbf{C}}\cdikk~to $[S_2]$ in $S$.
  \end{enumerate}
\end{proposition}

\begin{example}
  We go back to sequent (i') in Example \ref{ex:5}. 
  Let 
  $S=\square(p\vee q)\Rightarrow \square r\supset\square s,\langle\square(p\vee q),\square r\Rightarrow \square s,[p\vee q,p,r\Rightarrow s]\rangle,[\Rightarrow s]$ and $S_1= \square(p\vee q),\square r\Rightarrow \square s,[p\vee q,p,r\Rightarrow s]$, $S_2= \ \Rightarrow s$, $T= p\vee q,p,r\Rightarrow s$. 
  Since $[S_2]$ is produced by ($\text{inter}_{\downarrow}$) from $T$, we have $S_2=\ \Rightarrow_T s$. We are intended to realize the block $[S_2]$ in $S$ by the tracking record of $Ant(S)$ in $S_1$. By definition, we have $\mathfrak{G}_{S_1}(S_1,Ant(S))=Ant(S)=\{\square(p\vee q)\}$ and $\mathfrak{G}_{S_1}(T,Ant(S))=\{p\vee q,p\}$.
  According to realization, by applying  $f_{S_1}(\cdot)$ to  $S_2$, we get $f_{S_1}(\Rightarrow_{T} s)=p\vee q,p \Rightarrow s$. 
  Thus, the entire output sequent is $\square(p\vee q)\Rightarrow \square r\supset\square s,\langle\square(p\vee q),\square r\Rightarrow \square s,[p\vee q,p,r\Rightarrow s]\rangle,[p\vee q,p\Rightarrow s]$. 
  And this is just (iii) in Example \ref{ex:5}, which is the right expansion of (i').
\end{example}

In order to define the proof-search procedure, we first divide all the rules of \textbf{C}\cdikk~into four groups as (R1):  all propositional and modal rules except $(\supset_R)$; (R2): (trans) and $(\text{inter}_{\rightarrow})$; (R3): $(\supset_R)$; (R4): $(\text{inter}_{\downarrow})$.


Let $S=\Gamma\Rightarrow\Delta$, we denote by  $\bar{\Delta}$ the sequent obtained by removing all (nested) occurrences of $\langle\cdot\rangle$-blocks in $\Delta$. \footnote{E.g., let  $\Delta = B, \langle \Sigma \Rightarrow \Pi\rangle, [\Lambda \Rightarrow [D\Rightarrow  E, \langle P\Rightarrow Q\rangle]$, then $\bar{\Delta} =  B, [\Lambda \Rightarrow [D\Rightarrow  E]]$.}

\begin{definition}[Saturation]\label{3-saturation}
  Let $S =\Gamma\Rightarrow\Delta$ be a sequent and not an axiom. $S$ is called: 
  \begin{itemize}
    \item {\rm R1-saturated} if $\Gamma\Rightarrow\bar{\Delta}$ satisfies all the saturation conditions of R1 rules;
    \item {\rm R2-saturated} if $S$ is R1-saturated and $S$ satisfies saturation conditions of R2 rules for blocks $\langle S_1\rangle,[S_2]$ s.t. $S_1 \in^{\langle\cdot\rangle}_0 S$ and $S_2 \in^{[\cdot]}_0 S$;
    \item {\rm R3-saturated} if $S$ is R2-saturated and $S$ satisfies saturation conditions of R3 rules for formulas $A\supset B\in \Delta$;
    \item {\rm R4-saturated} $S$ is R3-saturated and $S$ satisfies saturation conditions of R4 rule for each implication block $\langle\Sigma\Rightarrow\Pi,[S_1]\rangle$ s.t. $\Sigma\Rightarrow\Pi,[S_1]\in^{\langle\cdot\rangle}_0 S$.
  \end{itemize}
\end{definition}

\begin{definition}[Global saturation]
  Let $S$ be a sequent and not an axiom. $S$ is called {\rm global-R$i$-saturated} if for each $T\in^+S$, $T$ is Ri-saturated where $i\in\{1,2,3\}$; {\rm global-saturated} if for each $T\in^+S$, $T$ is R4-saturated.
\end{definition}

In order to specify the proof-search procedure, we make use of the following four macro-steps that extend a given derivation $\+D$ by expanding a leaf $S$. Each procedure applies rules {\em non-redundantly} to some $T=\Gamma\Rightarrow\Delta\in^+S$. 
\begin{itemize}
  \item $\textbf{EXP1}(\+D, S, T)= \+D'$ where $\+D'$ is the extension of $\+D$  obtained by applying  R1-rules to every formula in $\Gamma\Rightarrow\bar{\Delta}$. 
  \item $\textbf{EXP2}(\+D, S, T)= \+D'$ where $\+D'$ is the extension of $\+D$  obtained by applying R2-rules to blocks $\langle T_i\rangle,[T_j]\in\Delta$. 
  \item $\textbf{EXP3}(\+D, S, T)= \+D'$ where $\+D'$ is the extension of $\+D$  obtained by applying R3-rules to formulas $A\supset B\in\Delta$. 
  \item  $\textbf{EXP4}(\+D, S)= \+D'$ where $\+D'$ is the extension of $\+D$  obtained by applying (i) R4-rule to each implication block $T'\in^+S$ and (ii) realization procedures to modal blocks produced in (i). This step extends $\+D$ by a \emph{single} branch whose leaf is denoted by $S'$. 
\end{itemize}

It can be proved that each of these four macro-steps terminates. The claim is almost obvious except for $\textbf{EXP1}$. (for $\textbf{EXP2}$ and $\textbf{EXP3}$, see \cite[Proposition 46]{FIK-csl2024}).

\begin{proposition}\label{prop:expi-terminates}
  Given a finite derivation $\+D$, a finite leaf $S$ of $\+D$ and $T\in^+S$, then for $i\in\{1,2,3,4\}$, each $\textbf{EXPi}(\+D, S, T)$ terminates by producing a finite expansion of $\+D$ where all sequents are finite. 
\end{proposition}

We first give the procedure $\text{PROCEDURE}_0(S_0)$ which builds a derivation with root $S_0$ and only uses the macro-steps $\textbf{EXP1}(\cdot)$ to $\textbf{EXP3}(\cdot)$, thus only the rule of  \local. Therefore  $\text{PROCEDURE}_0(\Rightarrow A)$ decides whether a formula $A$ is valid in \dik. In addition the procedure $\text{PROCEDURE}_0(\cdot)$ is then used as a subroutine of $\text{PROCEDURE}(\Rightarrow A)$ to obtain either a proof of $A$ or a global-saturated sequent, see Algorithm \ref{PROCEDURE}.

\begin{proposition}\label{proc0term}
	Given a sequent $S_0$, ${\rm PROC_0}(S_0)$ produces a finite derivation with all the leaves axiomatic or at least one global-R3-saturated leaf.
\end{proposition}

\vspace{-.5cm}
\begin{algorithm}[!h]
  \caption{$\text{PROCEDURE}_0(S_0)$}\label{PROCEDURE0}
  \KwIn{$S_0$}
  initialization $\+D= \ \Rightarrow S_0$\;
    \Repeat{FALSE}{
      \uIf{all the leaves of $\+D$ are axiomatic}{
        return ``PROVABLE" and $\+D$ }
      \uElseIf{there is a non-axiomatic leaf of $\+D$ which is global-R3-saturated}{return $\+D$} 
    \Else{{\bf select} one non-axiomatic leaf $S$ of $\+D$ that is not global-R3-saturated\\
         { \uIf{ $S$ is global-R2-saturated}{{\bf for} all $T\in^+S$ that is not R3-saturated, let $\+D = \textbf{EXP3}(\+D, S, T)$}
        \uElseIf{$S$ is global-R1-saturated}
            {{\bf for} all  $T\in^+ S$ that is not R2-saturated, let $\+D = \textbf{EXP2}(\+D, S, T)$}
         \uElse{{\bf for} all  $T\in^+ S$ that is not R1-saturated, let $\+D = \textbf{EXP1}(\+D, S, T)$}
          }
        }
    }
\end{algorithm}



\begin{algorithm}[H]
	\caption{$\text{PROCEDURE}(A)$}\label{PROCEDURE}
	\KwIn{$A$}
	initialization $\+D= \text{PROCEDURE}_0 (\Rightarrow A)$\;
	\uIf{all the leaves of $\+D$ are axiomatic}{
		return ``PROVABLE'' and $\+D$ }
	\uElse{
		\While{(No global saturated leaf of  $\+D$ is found)}{
		{\bf select} one  global-R3-saturated leaf $S$ of $\+D$   \\
		let $\+D =\textbf{EXP4}(\+D, S)$\\
		let $S'$ be the leaf of the unique branch of $\+D$  expanded by $\textbf{EXP4}(\+D, S)$
	 extend $\+D$ by applying $\text{PROCEDURE}_0(S')$\\
		}
		return  ``UNPROVABLE" and $\+D$
		}
\end{algorithm}

Lastly, we show that $\text{PROCEDURE}(A)$ terminates. 

\begin{theorem}[Termination for \textbf{C}\cdikk]\label{termination}
  Let $A$ be a formula. Proof-search for $\Rightarrow A$ in {\rm \textbf{C}}\cdikk~terminates with a finite derivation in which either all the leaves are axiomatic or there is at least one global-saturated leaf.
\end{theorem}

We can obtain decision procedures  for  \cdid~and \cdit~too in a similar way: we consider a cumulative version {\rm \textbf{C}\cdid}~and {\rm \textbf{C}\cdit} of the respective calculi and we define suitable saturation conditions, for  a sequent $S$: \\
(\axiomd): if $\Gamma^\square\cup\Delta^\Diamond$ is non-empty.  then $\Delta$ is not $[\cdot]$-free. \\
($\axiomt_\square$/$\axiomt_\Diamond$): if $\square A\in \Gamma$ (resp. $\Diamond A\in\Delta$), then $A\in\Gamma$ (resp. $A\in\Delta$).

The saturation condition for (\axiomd) prevents  a useless generation of infinitely nested  empty blocks $[ \Rightarrow [\ldots  \Rightarrow [ \Rightarrow ]\ldots]]$ (we call it an empty structure) by the backward application of the (\axiomd)-rule. 

The procedure $\mathrm{PROCEDURE}_0$ integrates  the rules for (\axiomd) or (\axiomt)'s according to the logic: the rule (\axiomd) is applied immediately after each round of $\textbf{EXP2}(\cdot)$ while the two (\axiomt) rules integrated in  $\textbf{EXP1}(\cdot)$.
We can obtain: 

\remove{
As for the procedure, we treat the rules of (\axiomd) and (\axiomt)'s differently. As extra steps in $\mathrm{PROCEDURE}_0$ (see Algorithm \ref{PROCEDURE0}), the rule (\axiomd) is applied immediately after each round of $\textbf{EXP2}(\cdot)$ while the two (\axiomt) rules are regarded as part of the basic rules and applied within $\textbf{EXP1}(\cdot)$ together with other logical rules. Different levels of saturation in \textbf{C}\cdid~and \textbf{C}\cdit~are then adjusted accordingly. Termination for the procedures of both \textbf{C}\cdid~and \textbf{C}\cdit~are direct. Following the same proof for \textbf{C}\cdikk, we can obtain termination of proof-search in these two calculi as well.
}
\begin{theorem}[Termination for {\rm \textbf{C}\cdid}~and {\rm \textbf{C}\cdit}]
  Let $A$ be a formula. Proof-search for the sequent $\Rightarrow A$ in {\rm \textbf{C}\cdid}~and {\rm \textbf{C}\cdit}~terminates with a finite derivation in which either all the leaves are axiomatic or there is at least one global-saturated leaf.
\end{theorem}

%% file: 04-completeness.tex
\section{Completeness}\label{sec:completeness}

Using the decision procedure of the previous section, we show how to build a countermodel for an unprovable formula, which entails the completeness of \textbf{C}\cdikk. The construction is then adapted to \textbf{C}\cdid~and \textbf{C}\cdit. 

Given a global-saturated sequent $S$ in \textbf{C}\cdikk, define a model $\+M_S$ for it as:

\begin{definition}\label{def:model-construction}
  The model $\+M_S=(W_S,\leq_S,R_S,V_S)$ is a quadruple where\\
  - $W_S=\{x_{\Phi\Rightarrow\Psi}~|~\Phi\Rightarrow\Psi\in^+S\}$; \\
  - $x_{S_1}\leq_S x_{S_2}$ if $S_1\subseteq^{\-S} S_2$; \qquad $R_Sx_{S_1}x_{S_2}$ if $S_2\in^{[\cdot]}_0S_1$; \\
  - for each $p\in\mathbf{At}$, let $V_S(p)=\{x_{\Phi\Rightarrow\Psi}~|~p\in \Phi\}$.
\end{definition}

\begin{proposition}\label{prop:hp-fc-property}
  $\+M_S$ satisfies (FC) and (DC).
\end{proposition}

\begin{lemma}[Truth Lemma for {\rm \textbf{C}}\cdikk]\label{lem:truth-lemma-cdik}
  Let $S$ be a global-saturated sequent in {\rm \textbf{C}}\cdikk~and $\+M_{S}=(W_S,\leq_S,R_S,V_S)$ defined as above. (a). If $A\in\Phi$, then $\+M_{S},x_{\Phi\Rightarrow\Psi}\Vdash A$; (b). If $A\in\Psi$, then $M_{S},x_{\Phi\Rightarrow\Psi}\nVdash A$.
\end{lemma}

\noindent By truth lemma we obtain as usual the completeness of {\rm \textbf{C}}\cdikk.

\begin{theorem}[Completeness of \textbf{C}\cdikk] If $A$ is valid in \dik, then $A$ is provable in 
\cdikk.
\end{theorem}

\begin{example}
	We show how to build a countermodel for the formula $(\Diamond p \supset \Box q) \supset \Box (p \supset q)$ which is not provable in \textbf{C}\cdikk. Ignoring the first step, we initialize the derivation with $\Diamond p \supset \Box q \Rightarrow \Box (p \supset q)$. By backward application of rules, one branch of the derivation ends up with the following saturated sequent $S_0=\Diamond p\supset\square q\Rightarrow\square(p\supset q),\Diamond p,[\Rightarrow p\supset q,p,\langle p\Rightarrow q\rangle]$, 
  and we further let 
  $S_1= \ \Rightarrow p\supset q,p,\langle p\Rightarrow q\rangle$ while 
  $S_2=p\Rightarrow q$. 
	We then get the model $M_{S_0} = (W, \leq, R, V)$ where  $W=\{x_{S_0}, x_{S_1}, x_{S_2} \}$, $x_{S_1} \leq x_{S_2}$, $Rx_{S_0}x_{S_1}$, 
	$V(p)=\{x_{S_2}\}$ and $V(q)=\varnothing$.
  It is easy to see that $x_{S_0} \not\Vdash (\Diamond p \supset \Box q) \supset \Box (p \supset q)$.	
\end{example}

Next, we consider the completeness of \textbf{C}\cdid~and \textbf{C}\cdit. We consider the model $\+M_S=(W_S,\leq_S,R_S,V_S)$ for a global-saturated sequent $S$ in either calculi, where $W_S,\leq_S$ and $V_S$ as in Definition \ref{def:model-construction}, $R_S$ modified as follows: \\
- For \textbf{C}\cdid: $R_Sx_{S_1}x_{S_2}$ if $S_2\in^{[\cdot]}_0S_1$ or $Suc(S_1)$ is $[\cdot]$-free and $x_{S_1} = x_{S_2}$;\\ 
- For \textbf{C}\cdit: $R_Sx_{S_1}x_{S_2}$ if $S_2\in^{[\cdot]}_0S_1$ or $x_{S_1} = x_{S_2}$.

\noindent Trivially the relation $R_S$ is serial or reflexive according to  \cdid \ or \cdit,   moreover models for \textbf{C}\cdid~and \textbf{C}\cdit~still satisfy (FC) and (DC). Finally,

\remove{
\begin{proposition}\label{prop:r-dit-dit}
  If the root $S$ is a global-saturated sequent in {\rm \textbf{C}}\cdid (resp. {\rm \textbf{C}}\cdit), then $R_S$ is serial (resp. reflexive).
\end{proposition}
We can extend the truth lemma to the models built for \textbf{C}\cdid~and \textbf{C}\cdit. So that we finally obtain:
}

\begin{theorem}[Completeness of \textbf{C}\cdid~and \textbf{C}\cdit] 
  If $A$ is valid in \did~ (resp. \dit), then $A$ is provable in  {\rm \textbf{C}}\cdid \ (resp. {\rm \textbf{C}}\cdit).
\end{theorem}

%% file: 05-conclusion.tex
\section{Conclusion}\label{sec:conclusion}
We studied \dik, the basic intuitionistic modal logic with locally defined modalities and some of its extensions.  
In further research, we intend to investigate the extension of both axiomatization and calculi to the whole modal cube.
For instance, we would like to provide a (terminating) calculus for the \textbf{S4} extension of \dik~(studied in \cite{Balbiani:et:al:2021}).
Since \dik~is incomparable with \ik, we may wonder what is the ``super'' intuitionistic modal logic obtained by combining both is.
Our broader goal is to build a framework of axiomatization and uniform calculi for a wide range of \textbf{IML}s, including other natural variants that have been little studied or remain entirely unexplored so far.

%% file: A-sec2-proofs.tex
\section{Proofs in Section \ref{sec:logic}}
%
%
{\noindent \textbf{Lemma~\ref{lemma:HB:monotonicity}.}}
{\it Let $(W,{\leq},{R},V)$ be a forward and downward confluent model.
For all $A{\in}\Fo$ and for all $s,t{\in}W$, if $s{\Vdash}A$ and $s{\leq}t$ then $t{\Vdash}A$.}
\begin{proof}
%
%
By induction on $A$.\qed
\end{proof}
%
%
%
%
{\noindent \textbf{Lemma~\ref{lemma:about:D:and:axioms}.}}
{\it If $\AxiomSer{\in}X$ or $\AxiomRef{\in}X$ then $\square p{\supset}\lozenge p$ and $\neg\square\bot$ are in \dik$X$.}
\begin{proof}
By using $\K_{\square}$, \textbf{NEC}, $\K_{\Diamond}$ and $\NotPosBot$.\qed
\end{proof}
%
%
{\noindent \textbf{Theorem~\ref{theorem:soundness:dfik}.}}
{\it \dik$X$-derivable formulas are ${\mathcal C}_{\fdcfra}^{X}$-validities.}
\begin{proof}
By induction on the length of a derivation of $A$.\qed
\end{proof}
%
%
{\noindent \textbf{Lemma~\ref{lemma:another:property:of:LIKD}.}}
{\it If $\AxiomSer{\in}X$ or $\AxiomRef{\in}X$ then for all theories $\Gamma$, $\lozenge\square\Gamma{\subseteq}\Gamma$.}
\begin{proof}
Suppose $\AxiomSer{\in}X$.
Let$\Gamma$ be a theory.
If $\lozenge\square\Gamma{\not\subseteq}\Gamma$ then there exists a formula $A$ such that $A{\in}\square\Gamma$ and $\lozenge A{\not\in}\Gamma$.
Hence, $\square A{\in}\Gamma$.
Since $\AxiomSer{\in}X$, therefore by Lemma~\ref{lemma:about:D:and:axioms}, $\square A\supset\lozenge A$ is in \dik{$X$}.
Since $\square A{\in}\Gamma$, therefore $\lozenge A{\in}\Gamma$: a contradiction.\qed
\end{proof}
%
%
%
%
%
%
%
%
%
%
%
%
%
%
%
%
%
%
{\noindent \textbf{Lemma~\ref{fc:canonical:frame:is:forward:confluent}.}}
{\it
\begin{enumerate}
\item $(W_{\mathbf{L}},\leq_{\mathbf{L}},R_{\mathbf{L}})$ is forward confluent,
\item $(W_{\mathbf{L}},\leq_{\mathbf{L}},R_{\mathbf{L}})$ is downward confluent,
\item if $\AxiomSer{\in}X$ (resp. $\AxiomRef{\in}X$, $\AxiomTra{\in}X$) then $(W_{\mathbf{L}},\leq_{\mathbf{L}},R_{\mathbf{L}})$ is serial (resp. reflexive, transitive).
\end{enumerate}
}
\begin{proof}
$\mathbf{(1)}$~See~\cite{FIK-csl2024}.
\\
\\
$\mathbf{(2)}$~Let $\Gamma,\Delta,\Lambda{\in}W_{\mathbf{L}}$ be such that $\Gamma\leq_{\mathbf{L}}\Delta$ and ${R_{\mathbf{L}}}\Delta\Lambda$.
Hence, $\Gamma{\subseteq}\Delta$ and ${\square}\Delta{\subseteq}\Lambda$ and ${\lozenge}\Lambda{\subseteq}\Delta$.
\\
\\
We claim that ${\square}\Gamma{\subseteq}\Lambda$.
If not, there exists a formula $A$ such that $A{\in}{\square}\Gamma$ and $A{\not\in}\Lambda$.
Thus, ${\square}A{\in}\Gamma$.
Since $\Gamma{\subseteq}\Delta$, then ${\square}A{\in}\Delta$.
Since ${R_{\mathbf{L}}}\Delta\Lambda$, then $A{\in}\Lambda$: a contradiction.
Consequently, ${\square}\Gamma{\subseteq}\Lambda$.
\\
\\
We claim that for all formulas $A,B$, if ${\lozenge}A{\not\in}\Gamma$ and $A{\vee}B{\in}{\square}\Gamma$ then $B{\in}\Lambda$.
If not, there exists formulas $A,B$ such that ${\lozenge}A{\not\in}\Gamma$, $A{\vee}B{\in}{\square}\Gamma$ and $B{\not\in}\Lambda$.
Hence, ${\square}(A{\vee}B){\in}\Gamma$.
Thus, using the fact that $(\CD){\in}\mathbf{L}$, ${\lozenge}A{\vee}{\square}B{\in}\Gamma$.\footnote{This is our only use of axiom $(\CD)$ in the completeness proof.} 
Consequently, either ${\lozenge}A{\in}\Gamma$, or ${\square}B{\in}\Gamma$.
Since ${\lozenge}A{\not\in}\Gamma$, then ${\square}B{\in}\Gamma$.
Since $\Gamma{\subseteq}\Delta$, then ${\square}B{\in}\Delta$.
Since ${R_{\mathbf{L}}}\Delta\Lambda$, then $B{\in}\Lambda$: a contradiction.
Hence, for all formulas $A,B$, if ${\lozenge}A{\not\in}\Gamma$ and $A{\vee}B{\in}{\square}\Gamma$ then $B{\in}\Lambda$.
\\
\\
Let ${\mathcal S}{=}\{\Theta : \Theta$ is a theory such that {\bf (1)}~${\square}\Gamma{\subseteq}\Theta$, {\bf (2)}~$\Theta{\subseteq}\Lambda$ and {\bf (3)}~for all formulas $A,B$, if ${\lozenge}A{\not\in}\Gamma$ and $A{\vee}B{\in}\Theta$ then $B{\in}\Lambda\}$.
\\
\\
Obviously, ${\square}\Gamma\in{\mathcal S}$.
Thus, ${\mathcal S}$ is nonempty.
Moreover, for all nonempty chains $(\Theta_{i})_{i{\in}I}$ of elements of ${\mathcal S}$, $\bigcup\{\Theta_{i} : i{\in}I\}$ is an element of ${\mathcal S}$.
Consequently, by Zorn's Lemma, ${\mathcal S}$ possesses a maximal element $\Theta$.
Hence, $\Theta$ is a theory such that ${\square}\Gamma{\subseteq}\Theta$, $\Theta{\subseteq}\Lambda$ and for all formulas $A,B$, if ${\lozenge}A{\not\in}\Gamma$ and $A{\vee}B{\in}\Theta$ then $B{\in}\Lambda$.
\\
\\
Thus, it only remains to be proved that $\Theta$ is proper and prime and ${R_{\mathbf{L}}}\Gamma\Theta$.
\\
\\
We claim that $\Theta$ is proper.
If not, ${\bot}{\in}\Theta$.
Since $\Theta{\subseteq}\Lambda$, then ${\bot}{\in}\Lambda$: a contradiction.
Consequently, $\Theta$ is proper.
\\
\\
We claim that $\Theta$ is prime.
If not, there exists formulas $A,B$ such that $A{\vee}B{\in}\Theta$, $A{\not\in}\Theta$ and $B{\not\in}\Theta$.
Hence, by the maximality of $\Theta$ in ${\mathcal S}$, $\Theta{+}A{\not\in}{\mathcal S}$ and $\Theta{+}B{\not\in}{\mathcal S}$.
Thus, either there exists a formula $C$ such that $C{\in}\Theta{+}A$ and $C{\not\in}\Lambda$, or there exists formulas $C,D$ such that ${\lozenge}C{\not\in}\Gamma$, $C{\vee}D{\in}\Theta{+}A$ and $D{\not\in}\Lambda$ and either there exists a formula $E$ such that $E{\in}\Theta{+}B$ and $E{\not\in}\Lambda$, or there exists formulas $E,F$ such that ${\lozenge}E{\not\in}\Gamma$, $E{\vee}F{\in}\Theta{+}B$ and $F{\not\in}\Lambda$.
Consequently, we have to consider the following four cases.
\\
$\mathbf{(1)}$ Case ``there exists a formula $C$ such that $C{\in}\Theta{+}A$ and $C{\not\in}\Lambda$ and there exists a formula $E$ such that $E{\in}\Theta{+}B$ and $E{\not\in}\Lambda$'':
Hence, $A{\supset}C{\in}\Theta$ and $B{\supset}E{\in}\Theta$.
Thus, using axioms and inference rules of Intuitionistic Propositional Logic, $A{\vee}B{\supset}C{\vee}E{\in}\Theta$.
Since $A{\vee}B{\in}\Theta$, then $C{\vee}E{\in}\Theta$.
Since $\Theta{\subseteq}\Lambda$, then $C{\vee}E{\in}\Lambda$.
Since $C{\not\in}\Lambda$ and $E{\not\in}\Lambda$, then $C{\vee}E{\not\in}\Lambda$: a contradiction.
\\
$\mathbf{(2)}$ Case ``there exists a formula $C$ such that $C{\in}\Theta{+}A$ and $C{\not\in}\Lambda$ and there exists formulas $E,F$ such that ${\lozenge}E{\not\in}\Gamma$, $E{\vee}F{\in}\Theta{+}B$ and $F{\not\in}\Lambda$'':
Consequently, $A{\supset}C{\in}\Theta$ and $B{\supset}E{\vee}F{\in}\Theta$.
Hence, using axioms and inference rules of Intuitionistic Propositional Logic, $A{\vee}B{\supset}E{\vee}C{\vee}F{\in}\Theta$.
Since $A{\vee}B{\in}\Theta$, then $E{\vee}C{\vee}F{\in}\Theta$.
Since ${\lozenge}E{\not\in}\Gamma$, then $C{\vee}F{\in}\Lambda$.
Since $C{\not\in}\Lambda$ and $F{\not\in}\Lambda$, then $C{\vee}F{\not\in}\Lambda$: a contradiction.
\\
$\mathbf{(3)}$ Case ``there exists formulas $C,D$ such that ${\lozenge}C{\not\in}\Gamma$, $C{\vee}D{\in}\Theta{+}A$ and $D{\not\in}\Lambda$ and there exists a formula $E$ such that $E{\in}\Theta{+}B$ and $E{\not\in}\Lambda$'':
Thus, $A{\supset}C{\vee}D{\in}\Theta$ and $B{\supset}E{\in}\Theta$.
Consequently, using axioms and inference rules of Intuitionistic Propositional Logic, $A{\vee}B{\supset}C{\vee}D{\vee}E{\in}\Theta$.
Since $A{\vee}B{\in}\Theta$, then $C{\vee}D{\vee}E{\in}\Theta$.
Since ${\lozenge}C{\not\in}\Gamma$, then $D{\vee}E{\in}\Lambda$.
Since $D{\not\in}\Lambda$ and $E{\not\in}\Lambda$, then $D{\vee}E{\not\in}\Lambda$: a contradiction.
\\
$\mathbf{(4)}$ Case ``there exists formulas $C,D$ such that ${\lozenge}C{\not\in}\Gamma$, $C{\vee}D{\in}\Theta{+}A$ and $D{\not\in}\Lambda$ and there exists formulas $E,F$ such that ${\lozenge}E{\not\in}\Gamma$, $E{\vee}F{\in}\Theta{+}B$ and $F{\not\in}\Lambda$'':
Hence, $A{\supset}C{\vee}D{\in}\Theta$ and $B{\supset}E{\vee}F{\in}\Theta$.
Thus, using axioms and inference rules of Intuitionistic Propositional Logic, $A{\vee}B{\supset}C{\vee}E{\vee}D{\vee}F{\in}\Theta$.
Since $A{\vee}B{\in}\Theta$, then $C{\vee}E
$\linebreak$
{\vee}D{\vee}F{\in}\Theta$.
Since ${\lozenge}C{\not\in}\Gamma$ and ${\lozenge}E{\not\in}\Gamma$, then using axiom $(\DP)$, ${\lozenge}(C{\vee}E){\not\in}\Gamma$.
Since $C{\vee}E{\vee}D{\vee}F{\in}\Theta$, then $D{\vee}F{\in}\Lambda$.
Since $D{\not\in}\Lambda$ and $F{\not\in}\Lambda$, then $D{\vee}F{\not\in}\Lambda$: a contradiction.
\\
Consequently, $\Theta$ is prime.
\\
\\
We claim that ${R_{\mathbf{L}}}\Gamma\Theta$.
If not, there exists a formula $A$ such that $A{\in}\Theta$ and ${\lozenge}A{\not\in}\Gamma$.
Hence, using axioms and inference rules of Intuitionistic Propositional Logic, $A{\vee}{\bot}{\in}\Theta$.
Since ${\lozenge}A{\not\in}\Gamma$, then ${\bot}{\in}\Lambda$: a contradiction.
Thus, ${R_{\mathbf{L}}}\Gamma\Theta$.
\\
\\
All in all, we have proved that $\Theta{\in}W_{\mathbf{L}}$ is such that $R_{\mathbf{L}}\Gamma\Theta$ and $\Theta\leq_{\mathbf{L}}\Lambda$.
\\
\\
$\mathbf{(3)}$~Suppose $\AxiomSer{\in}X$.
We demonstrate $(W_{\mathbf{L}},\leq_{\mathbf{L}},R_{\mathbf{L}})$ is serial.
Let $\Gamma{\in}W_{\mathbf{L}}$.
Let ${\mathcal S}{=}\{\Delta$: $\Delta$ is a theory such that $\square\Gamma{\subseteq}\Delta$, $\lozenge\Delta{\subseteq}\Gamma$ and $\bot{\not\in}\Delta\}$.
Since $\AxiomSer{\in}X$, therefore by Lemmas~\ref{lemma:about:D:and:axioms} and~\ref{lemma:another:property:of:LIKD}, $\square\bot{\not\in}\Gamma$ and $\lozenge\square\Gamma{\subseteq}\Gamma$.
Hence, $\square\Gamma{\in}{\mathcal S}$.
Thus, ${\mathcal S}$ is nonempty.
Moreover, for all $\subseteq$-chains $(\Delta_{i})_{i{\in}I}$ of elements of ${\mathcal S}$, $\bigcup_{i{\in}I}\Delta_{i}$ is in ${\mathcal S}$.
Consequently, by Zorn's Lemma, there exists a maximal elements $\Delta$ in ${\mathcal S}$.
By using \textbf{DP}, the reader may easily verify that $\Delta{\in}W_{\mathbf{L}}$.
Moreover, obviously, ${R_{\mathbf{L}}}\Gamma\Delta$.
\\
\\
Suppose $\AxiomRef{\in}X$.
By using $\AxiomRef$, the reader may easily verify that $(W_{\mathbf{L}},\leq_{\mathbf{L}},R_{\mathbf{L}})$ is reflexive.
\\
\\
Suppose $\AxiomTra{\in}X$.
By using $\AxiomTra$, the reader may easily verify that $(W_{\mathbf{L}},\leq_{\mathbf{L}},R_{\mathbf{L}})$ is transitive.\qed
\end{proof}
%
%
{\noindent \textbf{Lemma~\ref{lemma:prime:proper:for:implication}.}}
{\it
Let $\Gamma$ be a prime theory.
\begin{enumerate}
    \item If $B\supset C\not\in\Gamma$ then there exists a prime theory $\Delta$ such that $\Gamma\subseteq\Delta$, $B\in\Delta$ and $C\not\in\Delta$,
    \item if $\square B\not\in\Gamma$ then there exists a prime theory $\Delta$ such that $R_{\mathbf{L}}\Gamma\Delta$ and $B\not\in\Delta$,
    \item if $\Diamond B\in\Gamma$ then there exists a prime theory $\Delta$ such that $R_{\mathbf{L}}\Gamma\Delta$ and $B\in\Delta$.
\end{enumerate}
}
\begin{proof}
Let $\Gamma$ be a prime theory.
For~(1), see~\cite{FIK-csl2024}.
For~(2), suppose $\square B\not\in\Gamma$.
Hence, by~\cite{FIK-csl2024}, there exists prime theories $\Lambda,\Delta$ such that $\Gamma\subseteq\Lambda$, $R_{\mathbf{L}}\Lambda\Delta$ and $B\not\in\Delta$.
Thus, there exists a prime theory $\Theta$ such that $R_{\mathbf{L}}\Gamma\Theta$ and $\Theta\subseteq\Delta$.
Since $B\not\in\Delta$, $B\not\in\Theta$.
For~(3), see~\cite{FIK-csl2024}.\qed
\end{proof}
%
%
{\noindent \textbf{Lemma~\ref{lemma:truth:lemma}.}}
{\it For all formulas $A$ and for all $\Gamma\in W_{\mathbf{L}}$, $A\in\Gamma$ if and only if $\Gamma\Vdash A$.}
\begin{proof}
By induction on $A$.
The case when $A$ is an atom is by definition of $V_{\mathbf{L}}$.
The cases when $A$ is of the form $\bot$, $\top$, $B\wedge C$ and $B\vee C$ are as usual.
The cases when $A$ is of the form $B\supset C$, $\square B$ and $\Diamond B$ use the Existence Lemma.\qed
\end{proof}
%
%
{\noindent \textbf{Proposition~\ref{lemma:l:min:C1}.}}
{\it
\begin{enumerate}
\item \dik$X$ is conservative over $\IPL$,
\item \dik$X$~contains all substitution instances of $\IPL$ and is closed with respect to modus ponens,
\item \dik$X$~has the disjunction property if and only if $\AxiomSer{\in}X$ or $\AxiomRef{\in}X$,
\item the addition of the law of excluded middle to \dik$X$~yields modal logic $\K$,
\item $\square$ and $\lozenge$ are independent in \dik$X$.
\end{enumerate}
}
\begin{proof}
$\mathbf{(1)}$ By the fact that for all partial orders $(W,{\leq})$, $(W,{\leq},\mathbf{Id})$ validates \dik$X$ ($\mathbf{Id}$ denoting the identity relation on $W$).
\\
\\
$\mathbf{(2)}$ By definition of \dik$X$.
\\
\\
$\mathbf{(3)}$ From left to right, it suffices to notice that although ${\lozenge}{\top}{\vee}{\square}{\bot}$ is in \dik, neither ${\lozenge}{\top}$ is in \ditra, nor ${\square}{\bot}$ is in \ditra.
From right to left, for the sake of the contradiction, suppose $\AxiomSer{\in}X$ or $\AxiomRef{\in}X$ and \dik$X$~does not possess the disjunction property.
Hence, there exists formulas $A_{1},A_{2}$ such that $A_{1}{\vee}A_{2}$ is in \dik$X$~and neither $A_{1}$ is in \dik$X$, nor $A_{2}$ is in \dik$X$.
Thus, by the completeness of \dik$X$, there exists frames $(W_{1},{\leq_{1}},{R_{1}})$ and $(W_{2},{\leq_{2}},{R_{2}})$, there exists models ${\mathcal M}_{1}{=}(W_{1},{\leq_{1}},{R_{1}},V_{1})$ and ${\mathcal M}_{2}{=}(W_{2},{\leq_{2}},{R_{2}},V_{2})$ based on these frames, there exists $s_{1}{\in}W_{1}$ and there exists $s_{2}{\in}W_{2}$ such that $(W_{1},{\leq_{1}},{R_{1}})\Vdash$\dik$X$, $(W_{2},{\leq_{2}},{R_{2}})\Vdash$\dik$X$, ${\mathcal M}_{1},s_{1}{\not\Vdash}A$ and ${\mathcal M}_{2},s_{2}{\not\Vdash}A$.
Moreover, for all $i=1,2$, $R_{i}$ is serial (resp. reflexive, transitive) if $\AxiomSer{\in}X$ (resp. $\AxiomRef{\in}X$, $\AxiomTra{\in}X$).
Let $s$ be a new element and ${\mathcal M}{=}(W,{\leq},{R},V)$ be the model such that
\begin{itemize}
    \item $W{=}W_{1}{\cup}W_{2}{\cup}\{s\}$,
    \item $\leq$ is the least preorder on $W$ containing $\leq_{1}$, $\leq_{2}$, $\{s\}{\times}W_{1}$ and $\{s\}{\times}W_{2}$,
    \item $R$ is the least binary relation on $W$ containing $R_{1}$, $R_{2}$ and $\{(s,s)\}$,
    \item for all atoms $p$, $V(p){=}V_{1}(p){\cup}V_{2}(p)$.
\end{itemize}
Obviously, ${\mathcal M}$ is a forward and downward confluent model.
In other respect, $R$ is serial (resp. reflexive, transitive) if $\AxiomSer{\in}X$ (resp. $\AxiomRef{\in}X$, $\AxiomTra{\in}X$).
Moreover, ${\mathcal M},s_{1}{\not\Vdash}A_{1}$ and ${\mathcal M},s_{2}{\not\Vdash}A_{2}$.
Consequently, ${\mathcal M},s{\not\Vdash}A_{1}{\vee}A_{2}$: a contradiction with the fact that $A_{1}{\vee}A_{2}$ is in \dik$X$.
\\
\\
$\mathbf{(4)}$ As done in~\cite{FIK-csl2024} for \fik.
\\
\\
$\mathbf{(5)}$ As done in~\cite{FIK-csl2024} for \fik.\qed
\end{proof}

%% file: A-sec3-proofs.tex
\section{Proofs in Section \ref{sec:calculi}}

\begin{lemma}\label{sharp-preserving}
  Let $\Lambda\Rightarrow\Theta$ be a sequent, $\+M=(W,\leq,R,V)$ a model and $x\in W$. 
  If $x\nVdash \Theta$, then $x\nVdash\Theta^\flat$. And dually, if $x\Vdash \ \Rightarrow\Theta^\#$, then $x\Vdash \ \Rightarrow \Theta$. 
\end{lemma}

\begin{proof}
  Induction on the structure of $\Theta$.
\end{proof}

\begin{proposition}
    ($\text{inter}_{\rightarrow}$) is valid for semantics on (FC) frames and dually ($\text{inter}_{\downarrow}$) is valid for semantics on (DC) frames.
\end{proposition}
  
\begin{proof}
  The proof for (FC) can be found in \cite[Lemma 29]{FIK-csl2024}. For (DC), we prove the basic case here, i.e. context $G\{\}$ in the rule is empty. The general case can be shown by induction on the structure of $G\{\}$.
    
  Assume for the sake of a contradiction that the rule is not valid, then there is a model $\+M=(W,\leq,R,V)$ and $x\in W$ s.t. $x\Vdash \Gamma\Rightarrow\Delta,\langle\Sigma\Rightarrow\Pi,[\Lambda\Rightarrow\Theta]\rangle,[\Rightarrow\Theta^\#]$ and $x\nVdash \Gamma\Rightarrow\Delta,\langle\Sigma\Rightarrow\Pi,[\Lambda\Rightarrow\Theta]\rangle$. 
  Then we have (a). $x\nVdash\langle\Sigma\Rightarrow\Pi,[\Lambda\Rightarrow\Theta]\rangle$ and (b). $x\Vdash[\Rightarrow\Theta^\#]$. According to (a), there is some $x'\geq x$ s.t. $x'\nVdash\Sigma\Rightarrow\Pi,[\Lambda\Rightarrow\Theta]$. It follows that $x'\nVdash[\Lambda\Rightarrow\Theta]$, which means there is some $y$ with $Rx'y$ s.t. $y\Vdash\Lambda$ and $y\nVdash\Theta$. 
  Meanwhile, since $\+M$ satisfies (DC), by $x\leq x'$ and $Rx'y$, we see that there is some $y_0$ s.t. $Rxy_0$ and $y_0\leq y$. According to (b), $y_0\Vdash\Theta^\#$. By Lemma \ref{sharp-preserving}, $y_0\Vdash\Theta$. Since $y_0\leq y$, we see that $y\Vdash\Theta^\#$, a contradiction. 
  \qed
\end{proof}

{\noindent \textbf{Theorem~\ref{thm:soundness-cdik}.}} 
{\it If a sequent $S$ is provable in \cdikk, then it is a validity in \dik.}

\begin{proof}
  By verifying each rule in \cdikk~is valid in \dik.
\end{proof}

\begin{lemma}\label{Prop:empty-block-derivable}
  Let $S=\Gamma\Rightarrow\Delta$ be a sequent and $\Delta$ is not $[\cdot]$-free, $G\{\}$ a context. 
  Then $G\{\Gamma\Rightarrow\Delta\}$ is provable in \local~if and only if $G\{\Gamma\Rightarrow\Delta,[\Rightarrow]\}$ is provable. 
\end{lemma}

\begin{proof}
  The left-to-right direction is straightforward by applying $(w_R)$ to add an empty $[\Rightarrow]$. For the right-to-left, it can be done by induction on the structure of the derivation. \qed
\end{proof}
 
{\noindent \textbf{Proposition~\ref{prop:lik-eq-fragment}.}} 
{\it The $(\text{inter}_{\downarrow})$ rule is admissible in \local. Consequently, a sequent $S$ is provable in  \cdikk \ if and only if  $S$ is provable in \local.}

\begin{proof}
  We show the admissibility of ($\text{inter}_{\downarrow}$) by induction on the structure of a derivation in \local. Assume the premise $S=\Gamma\Rightarrow\Delta,\langle\Sigma\Rightarrow\Pi,[\Lambda\Rightarrow\Theta]\rangle,[\Rightarrow\Theta^\#]$ is derivable, 
  if it is an axiom, it is easy to see $\Gamma\Rightarrow\Delta,\langle\Sigma\Rightarrow\Pi,[\Lambda\Rightarrow\Theta]\rangle$ is also an axiom. 
  Otherwise, there is a derivation $\+D$ with $S$ be the conclusion. 
  
  Consider the last rule application $(r)$ in $\+D$ and we call the derivation above $(r)$ by $\+D_0$.
  
  \textbf{(i).} Assume the principal formula or block of $(r)$ is among $\Gamma,\Delta,\Sigma,\Pi$, which implies the $(r)$ application only concerns part(s) of $\Gamma,\Delta,\Sigma,\Pi$ while $\Lambda,\Theta,\Theta^\#$ are all side. Then it is easy to transform $\+D$ into a derivation for $\Gamma\Rightarrow\Delta,\langle\Sigma\Rightarrow\Pi,[\Lambda\Rightarrow\Theta]\rangle$. We only take one example here and other cases are similar. 
  Let $(r)$ be a single-premise rule only concerning $\Delta$ which means $\Gamma,\Sigma,\Pi$ are all side, then $\+D_0=\+D_0',(\Gamma\Rightarrow\Delta,\langle\Sigma\Rightarrow\Pi,[\Lambda\Rightarrow\Theta]\rangle,[\Rightarrow\Theta^\#])$ for some derivation $\+D_0'$. So $\+D$ is

  \[
    \AxiomC{$\+D_0'$}
    \UnaryInfC{$\Gamma\Rightarrow\Delta',\langle\Sigma\Rightarrow\Pi,[\Lambda\Rightarrow\Theta]\rangle,[\Rightarrow\Theta^\#]$}
    \RightLabel{($r$)}
    \UnaryInfC{$\Gamma\Rightarrow\Delta,\langle\Sigma\Rightarrow\Pi,[\Lambda\Rightarrow\Theta]\rangle,[\Rightarrow\Theta^\#]$}
    \DisplayProof
  \]

  and we can construct a derivation for $\Gamma\Rightarrow\Delta,\langle\Sigma\Rightarrow\Pi,[\Lambda\Rightarrow\Theta]\rangle$ as follows

  \[
    \AxiomC{$\+D_0'$}
    \UnaryInfC{$\Gamma\Rightarrow\Delta',\langle\Sigma\Rightarrow\Pi,[\Lambda\Rightarrow\Theta]\rangle,[\Rightarrow\Theta^\#]$}
    \RightLabel{(IH)}
    \UnaryInfC{$\Gamma\Rightarrow\Delta',\langle\Sigma\Rightarrow\Pi,[\Lambda\Rightarrow\Theta]\rangle$}
    \RightLabel{($r$)}
    \UnaryInfC{$\Gamma\Rightarrow\Delta,\langle\Sigma\Rightarrow\Pi,[\Lambda\Rightarrow\Theta]\rangle$}
    \DisplayProof
  \]

  \textbf{(ii).} Assume the principal formula or block of $(r)$ is in $\Theta^\#$ while all the $\Gamma,\Delta,\Sigma,\Pi,\Theta$ are side. In this case, $(r)$ can only be a right rule and $\+D_0=\+D_0',(\Gamma\Rightarrow\Delta,\langle\Sigma\Rightarrow\Pi,[\Lambda\Rightarrow\Theta]\rangle,[\Rightarrow(\Theta^\#)'])$ for some derivation $\+D_0'$. 
  Then $\+D$ is
  \[
    \AxiomC{$\+D_0'$}
    \UnaryInfC{$\Gamma\Rightarrow\Delta,\langle\Sigma\Rightarrow\Pi,[\Lambda\Rightarrow\Theta]\rangle,[\Rightarrow(\Theta^\#)']$}
    \RightLabel{($r$)}
    \UnaryInfC{$\Gamma\Rightarrow\Delta,\langle\Sigma\Rightarrow\Pi,[\Lambda\Rightarrow\Theta]\rangle,[\Rightarrow\Theta^\#]$}
    \DisplayProof
  \]

  We only show the case of $(\vee_R)$ here and other cases are similar. In this case, $(\vee_R)$ is applied to a sequent $\Rightarrow \Omega,A,B\in^{[\cdot]}\Rightarrow[\Rightarrow(\Theta^\#)']$. Hence $\Rightarrow(\Theta^\#)'$ can be rewritten explicitly as $\Rightarrow \Theta_0,[\Rightarrow\Psi,[\cdots[\Rightarrow\Omega,A,B]\cdots]]$ for some $\Theta_0$. Correspondingly, $\Lambda\Rightarrow\Theta$ is then of the form $\Lambda\Rightarrow \Theta'_0,[\Phi\Rightarrow\Psi,[\cdots[\Xi\Rightarrow\Omega',A\vee B]\cdots]]$. 
  And $\+D$ is presented explicitly as
  {
    \tiny
    \[
    \AxiomC{$\+D_0'$}
    \UnaryInfC{$\Gamma\Rightarrow\Delta,\langle\Sigma\Rightarrow\Pi,[\Lambda\Rightarrow \Theta'_0,[\Phi\Rightarrow\Psi,[\cdots[\Xi\Rightarrow\Omega',A\vee B]\cdots]]]\rangle,[\Rightarrow \Theta_0,[\Rightarrow\Psi,[\cdots[\Rightarrow\Omega,A,B]\cdots]]]$}
    \RightLabel{($\vee_R$)}
    \UnaryInfC{$\Gamma\Rightarrow\Delta,\langle\Sigma\Rightarrow\Pi,[\Lambda\Rightarrow \Theta'_0,[\Phi\Rightarrow\Psi,[\cdots[\Xi\Rightarrow\Omega',A\vee B]\cdots]]]\rangle,[\Rightarrow \Theta_0,[\Rightarrow\Psi,[\cdots[\Rightarrow\Omega,A\vee B]\cdots]]]$}
    \DisplayProof
  \]
  }
  
  Since $(w)$ is hp-admissible in \local, we can construct a derivation for $\Gamma\Rightarrow\Delta,\langle\Sigma\Rightarrow\Pi,[\Lambda\Rightarrow\Theta]\rangle$ as follows 
  {
    \tiny
    \[
      \AxiomC{$\+D_0'$}
      \UnaryInfC{$\Gamma\Rightarrow\Delta,\langle\Sigma\Rightarrow\Pi,[\Lambda\Rightarrow \Theta'_0,[\Phi\Rightarrow\Psi,[\cdots[\Xi\Rightarrow\Omega',A\vee B]\cdots]]]\rangle,[\Rightarrow \Theta_0,[\Rightarrow\Psi,[\cdots[\Rightarrow\Omega,A,B]\cdots]]]$}
      \RightLabel{($w$)}
      \UnaryInfC{$\Gamma\Rightarrow\Delta,\langle\Sigma\Rightarrow\Pi,[\Lambda\Rightarrow \Theta'_0,[\Phi\Rightarrow\Psi,[\cdots[\Xi\Rightarrow\Omega',A,B,A\vee B]\cdots]]]\rangle,[\Rightarrow \Theta_0,[\Rightarrow\Psi,[\cdots[\Rightarrow\Omega,A,B,A\vee B]\cdots]]]$}
      \RightLabel{(IH)}
      \UnaryInfC{$\Gamma\Rightarrow\Delta,\langle\Sigma\Rightarrow\Pi,[\Lambda\Rightarrow \Theta'_0,[\Phi\Rightarrow\Psi,[\cdots[\Xi\Rightarrow\Omega',A,B,A\vee B]\cdots]]]\rangle$}
      \RightLabel{($\vee_R$)}
      \UnaryInfC{$\Gamma\Rightarrow\Delta,\langle\Sigma\Rightarrow\Pi,[\Lambda\Rightarrow \Theta'_0,[\Phi\Rightarrow\Psi,[\cdots[\Xi\Rightarrow\Omega',A\vee B,A\vee B]\cdots]]]\rangle$}
      \RightLabel{($c_R$)}
      \UnaryInfC{$\Gamma\Rightarrow\Delta,\langle\Sigma\Rightarrow\Pi,[\Lambda\Rightarrow \Theta'_0,[\Phi\Rightarrow\Psi,[\cdots[\Xi\Rightarrow\Omega',A\vee B]\cdots]]]\rangle$}
      \DisplayProof
    \]
  }
  
  and we see the conclusion is just $\Gamma\Rightarrow\Delta,\langle\Sigma\Rightarrow\Pi,[\Lambda\Rightarrow \Theta]\rangle$.

  ~

  \textbf{(iii).} Assume the principal formula or block of $(r)$ is in $\Lambda\Rightarrow\Theta$ while all the $\Gamma,\Delta,\Sigma,\Pi,\Theta^\#$ are side. 

  $\+D_0=\+D_0',(\Gamma\Rightarrow\Delta,\langle\Sigma\Rightarrow\Pi,[\Lambda'\Rightarrow\Theta']\rangle,[\Rightarrow\Theta^\#])$ for some derivation $\+D_0'$. 
  Then $\+D$ is
  \[
    \AxiomC{$\+D_0'$}
    \UnaryInfC{$\Gamma\Rightarrow\Delta,\langle\Sigma\Rightarrow\Pi,[\Lambda'\Rightarrow\Theta']\rangle,[\Rightarrow\Theta^\#]$}
    \RightLabel{($r$)}
    \UnaryInfC{$\Gamma\Rightarrow\Delta,\langle\Sigma\Rightarrow\Pi,[\Lambda\Rightarrow\Theta]\rangle,[\Rightarrow\Theta^\#]$}
    \DisplayProof
  \]

  If $(r)$ satisfies at least one of the following conditions: 
  \begin{itemize}
    \item the principal formula or block is not in the $(\Theta')^\#$-part of $\Lambda'\Rightarrow\Theta'$,
    \item it is a left rule among $(\vee_L)(\wedge_L)(\square_L)(\supset_L)$,
    \item it is a structural rule (trans)($\text{inter}_{\rightarrow}$),
  \end{itemize}
  then it is easy to verify that $(\Theta')^\#=\Theta^\#$. And we can construct a derivation for $\Gamma\Rightarrow\Delta,\langle\Sigma\Rightarrow\Pi,[\Lambda\Rightarrow\Theta]\rangle$ as follows
  \[
    \AxiomC{$\+D_0'$}
    \UnaryInfC{$\Gamma\Rightarrow\Delta,\langle\Sigma\Rightarrow\Pi,[\Lambda'\Rightarrow\Theta']\rangle,[\Rightarrow(\Theta')^\#]$}
    \RightLabel{(IH)}
    \UnaryInfC{$\Gamma\Rightarrow\Delta,\langle\Sigma\Rightarrow\Pi,[\Lambda'\Rightarrow\Theta']\rangle$}
    \RightLabel{($r$)}
    \UnaryInfC{$\Gamma\Rightarrow\Delta,\langle\Sigma\Rightarrow\Pi,[\Lambda\Rightarrow\Theta]\rangle$}
    \DisplayProof
  \] 
  Otherwise, $(r)$ is among all the right rules plus $(\Diamond_L)$ and the principal formula is in the $(\Theta')^\#$-part of $\Lambda'\Rightarrow\Theta'$. If $(r)$ is $(\Diamond_L)$, then $(r)$ is applied to some $\Phi\Rightarrow\Psi,[A\Rightarrow]\in^{[\cdot]} \ \Rightarrow[\Rightarrow[\Lambda'\Rightarrow\Theta']]$. Hence the block $[\Lambda'\Rightarrow\Theta']$ can be rewritten as $[\Lambda'\Rightarrow\Theta_0,[\cdots,[\Phi\Rightarrow\Psi,[A\Rightarrow]]]]$ for some $\Theta_0$ and $\Xi$ and accordingly, $\Rightarrow\Theta^\#$ is of the form $\Rightarrow\Theta_0^\#,[\cdots,[\Rightarrow\Psi^\#]]$.

  Then $\+D_0=\+D_0',(\Gamma\Rightarrow\Delta,\langle\Sigma\Rightarrow\Pi,[\Lambda'\Rightarrow\Theta_0,[\cdots,[\Phi\Rightarrow\Psi,[A\Rightarrow]]]]\rangle,[\Rightarrow\Theta_0^\#,[\cdots,[\Rightarrow\Psi^\#]]])$ and $\+D$ is the following
  \[
    \AxiomC{$\+D_0'$}
    \UnaryInfC{$\Gamma\Rightarrow\Delta,\langle\Sigma\Rightarrow\Pi,[\Lambda'\Rightarrow\Theta_0,[\cdots,[\Phi\Rightarrow\Psi,[A\Rightarrow]]]]\rangle,[\Rightarrow\Theta_0^\#,[\cdots,[\Rightarrow\Psi^\#]]]$}
    \RightLabel{($\Diamond_L$)}
    \UnaryInfC{$\Gamma\Rightarrow\Delta,\langle\Sigma\Rightarrow\Pi,[\Lambda'\Rightarrow\Theta_0,[\cdots,[\Phi,\Diamond A\Rightarrow\Psi]]]\rangle,[\Rightarrow\Theta_0^\#,[\cdots,[\Rightarrow\Psi^\#]]]$}
    \DisplayProof
  \]

  Since $(w)$ is hp-admissible in \local, we can construct a derivation for $\Gamma\Rightarrow\Delta,\langle\Sigma\Rightarrow\Pi,[\Lambda\Rightarrow\Theta]\rangle$ as below.
  \[
    \AxiomC{$\+D_0'$}
    \UnaryInfC{$\Gamma\Rightarrow\Delta,\langle\Sigma\Rightarrow\Pi,[\Lambda'\Rightarrow\Theta_0,[\cdots,[\Phi\Rightarrow\Psi,[A\Rightarrow]]]]\rangle,[\Rightarrow\Theta_0^\#,[\cdots,[\Rightarrow\Psi^\#]]]$}
    \RightLabel{($w$)}
    \UnaryInfC{$\Gamma\Rightarrow\Delta,\langle\Sigma\Rightarrow\Pi,[\Lambda'\Rightarrow\Theta_0,[\cdots,[\Phi\Rightarrow\Psi,[A\Rightarrow]]]]\rangle,[\Rightarrow\Theta_0^\#,[\cdots,[\Rightarrow\Psi^\#,[\Rightarrow]]]]$}
    \RightLabel{(IH)}
    \UnaryInfC{$\Gamma\Rightarrow\Delta,\langle\Sigma\Rightarrow\Pi,[\Lambda'\Rightarrow\Theta_0,[\cdots,[\Phi\Rightarrow\Psi,[A\Rightarrow]]]]\rangle$}
    \RightLabel{($\Diamond_L$)}
    \UnaryInfC{$\Gamma\Rightarrow\Delta,\langle\Sigma\Rightarrow\Pi,[\Lambda'\Rightarrow\Theta_0,[\cdots,[\Phi,\Diamond A\Rightarrow\Psi]]]\rangle$}
    \DisplayProof
  \]
  We see the conclusion is just $\Gamma\Rightarrow\Delta,\langle\Sigma\Rightarrow\Pi,[\Lambda\Rightarrow\Theta]\rangle$.

  If $(r)$ is a right rule, then $\Lambda=\Lambda'$ and $(r)$ is applied to some $T\in^{[\cdot]} \ \Rightarrow[\Lambda'\Rightarrow\Theta']$. It can be shown in a similar way as what we have done in \textbf{(iii)}.

  ~

  \textbf{(iv).} ($r$) is ($\text{inter}_{\rightarrow}$) and concerns both $\langle\Sigma\Rightarrow\Pi,[\Lambda\Rightarrow\Theta]\rangle$ and a modal block in $\Delta$. Let $\Delta=\Delta',[\Phi\Rightarrow\Psi]$ and $\+D_0=\+D_0',(\Gamma\Rightarrow\Delta',\langle\Sigma\Rightarrow\Pi,[\Lambda\Rightarrow\Theta],[\Phi\Rightarrow\Psi^\flat]\rangle,[\Phi\Rightarrow\Psi],[\Rightarrow\Theta^\#])$ for some derivation $\+D_0'$. Then $\+D$ is
  \[
    \AxiomC{$\+D_0'$}
    \UnaryInfC{$\Gamma\Rightarrow\Delta',\langle\Sigma\Rightarrow\Pi,[\Lambda\Rightarrow\Theta],[\Phi\Rightarrow\Psi^\flat]\rangle,[\Phi\Rightarrow\Psi],[\Rightarrow\Theta^\#]$}
    \RightLabel{($\text{inter}_{\rightarrow}$)}
    \UnaryInfC{$\Gamma\Rightarrow\Delta',\langle\Sigma\Rightarrow\Pi,[\Lambda\Rightarrow\Theta],\rangle,[\Phi\Rightarrow\Psi],[\Rightarrow\Theta^\#]$}
    \DisplayProof
  \]
  and we can construct a derivation for $\Gamma\Rightarrow\Delta,\langle\Sigma\Rightarrow\Pi,[\Lambda\Rightarrow\Theta]\rangle$ as follows
  \[
    \AxiomC{$\+D_0'$}
    \UnaryInfC{$\Gamma\Rightarrow\Delta',\langle\Sigma\Rightarrow\Pi,[\Lambda\Rightarrow\Theta],[\Phi\Rightarrow\Psi^\flat]\rangle,[\Phi\Rightarrow\Psi],[\Rightarrow\Theta^\#]$}
    \RightLabel{(IH)}
    \UnaryInfC{$\Gamma\Rightarrow\Delta',\langle\Sigma\Rightarrow\Pi,[\Lambda\Rightarrow\Theta],[\Phi\Rightarrow\Psi^\flat]\rangle,[\Phi\Rightarrow\Psi]$}
    \RightLabel{($\text{inter}_{\rightarrow}$)}
    \UnaryInfC{$\Gamma\Rightarrow\Delta',\langle\Sigma\Rightarrow\Pi,[\Lambda\Rightarrow\Theta]\rangle,[\Phi\Rightarrow\Psi]$}
    \DisplayProof
  \] 
  We see the conclusion is just $\Gamma\Rightarrow\Delta,\langle\Sigma\Rightarrow\Pi,[\Lambda\Rightarrow\Theta]\rangle$.

  ~

  \textbf{(v).} ($r$) is ($\text{inter}_{\rightarrow}$) and concerns both $[\Rightarrow\Theta^\#]$ and an implication block in $\Delta$. In this case, the derivation $\+D$ is of the same form of the example we presented in (i), thus we can construct a derivation for $\Gamma\Rightarrow\Delta,\langle\Sigma\Rightarrow\Pi,[\Lambda\Rightarrow\Theta]\rangle$ similarly.

  ~

  \textbf{(vi).} ($r$) is ($\text{inter}_{\rightarrow}$) and concerns both $\langle\Sigma\Rightarrow\Pi,[\Lambda\Rightarrow\Theta]\rangle$ and $[\Rightarrow\Theta^\#]$. In this case, $\+D_0=\+D_0',(\Gamma\Rightarrow\Delta,\langle\Sigma\Rightarrow\Pi,[\Lambda\Rightarrow\Theta],[\Rightarrow]\rangle,[\Rightarrow\Theta^\#])$ for some derivation $\+D_0'$. Then $\+D$ is
  \[
    \AxiomC{$\+D_0'$}
    \UnaryInfC{$\Gamma\Rightarrow\Delta,\langle\Sigma\Rightarrow\Pi,[\Lambda\Rightarrow\Theta],[\Rightarrow]\rangle,[\Rightarrow\Theta^\#]$}
    \RightLabel{($\text{inter}_{\rightarrow}$)}
    \UnaryInfC{$\Gamma\Rightarrow\Delta,\langle\Sigma\Rightarrow\Pi,[\Lambda\Rightarrow\Theta]\rangle,[\Rightarrow\Theta^\#]$}
    \DisplayProof
  \]
  and we can construct a derivation as follows
  \[
    \AxiomC{$\+D_0'$}
    \UnaryInfC{$\Gamma\Rightarrow\Delta,\langle\Sigma\Rightarrow\Pi,[\Rightarrow],[\Lambda\Rightarrow\Theta]\rangle,[\Rightarrow\Theta^\#]$}
    \RightLabel{(IH)}
    \UnaryInfC{$\Gamma\Rightarrow\Delta,\langle\Sigma\Rightarrow\Pi,[\Rightarrow],[\Lambda\Rightarrow\Theta]\rangle$}
    \DisplayProof
  \]
  By Lemma \ref{Prop:empty-block-derivable}, we see $\Gamma\Rightarrow\Delta,\langle\Sigma\Rightarrow\Pi,[\Lambda\Rightarrow\Theta]\rangle$ is derivable as well.
  \qed
\end{proof}

\begin{definition}[Empty structure]
    We inductively define an empty structure (of nested modal blocks) $\+E_n$ as below.
    \begin{itemize}
      \item $\+E_0=[\Rightarrow]$ is an empty structure;
      \item for $n\geq 1$, $\+E_n=[\Rightarrow \+E_{n-1}]$ where $\+E_{n-1}$ is an empty structure.
    \end{itemize}
\end{definition}  

\begin{proposition}\label{prop:empty-structure}
    If $\Gamma\Rightarrow\Delta,\+E_n$ is derivable where $\+E_n$ is an empty structure for some $n$, then $\Gamma\Rightarrow\Delta$ is also derivable. 
\end{proposition}
  
\begin{proof}
  Since $\+E_n$ is a finite structure which contains only empty modal blocks, $\Gamma\Rightarrow\Delta$ is derivable by applying finitely-many ($\axiomd$) rule to $\Gamma\Rightarrow\Delta,\+E_n$.
\end{proof}

{\noindent \textbf{Lemma~\ref{lemma:dp-cdid}.}} 
{\it Suppose $S =\ \Rightarrow A_1, \ldots, A_m, \langle G_1\rangle, \ldots, \langle G_n\rangle,[H_1],\ldots,[H_l]$ is provable in \cdid (resp. \cditt), where $A_i$'s are formulas, $G_j$ and $H_k$'s are sequents. 
Furthermore, each $H_k$ is of the form $\Rightarrow\Theta_k$ and for each sequent $T\in^{[\cdot]} H_k$, $T$ has an empty antecedent. We claim there is at least one $\Rightarrow A_i, \Rightarrow \langle G_j\rangle$ or $\Rightarrow[H_k]$ which is provable for some $i\leq m, j\leq n$ or $k\leq l$ in \cdid (resp. \cditt).}
  
\begin{proof}
  We show the case of \cdid~by double induction on the height of a proof of $S$ and the modal depth of $S$. 
  If $S$ is an axiom, then some $\Rightarrow\langle G_j\rangle$ or $\Rightarrow[H_k]$ must be an axiom. Otherwise $S$ it is obtained by applying a rule to some $A_i,\langle G_j\rangle$ or $[H_k]$.
  
  First, suppose $S$ is derived by applying a rule to $A_1$ (to simplify the index). We only illustrate the following three cases. 
  Let $A_1 = B\land C$ and $S$ is derived by
    \[
    \AxiomC{$\Rightarrow B, A_2, \ldots, A_m, \overline{\langle G_j\rangle},\overline{[H_k]}$}
    \AxiomC{$\Rightarrow C, A_2, \ldots, A_m, \overline{\langle G_j\rangle},\overline{[H_k]}$}
    \RightLabel{($\wedge_R$)}
    \BinaryInfC{$\Rightarrow B\wedge C, A_2, \ldots, A_m, \overline{\langle G_j\rangle},\overline{[H_k]}$}
    \DisplayProof
    \]
  By IH on the first premise, if either for some $i\geq 2$, sequent $\Rightarrow A_i$ is provable, or some $\langle G_j\rangle$ is provable, then we are done; otherwise $\Rightarrow B$ must be provable while all the other $A_i$ or $\langle G_j\rangle$ are not. In this case by IH on the second premise it follows that $\Rightarrow C$ must be provable; then we conclude $\Rightarrow B\wedge C$ is provable by applying $(\land_R)$. 
    
  Let $A_1 = \Box B$ and $S$ is derived by
    \[
    \AxiomC{$\Rightarrow A_2, \ldots, A_m, \overline{\langle G_j\rangle},\overline{[H_k]},[\Rightarrow B]$}
    \RightLabel{($\square_R$)}
    \UnaryInfC{$\Rightarrow \square B,A_2, \ldots, A_m, \overline{\langle G_j\rangle},\overline{[H_k]}$}
    \DisplayProof
    \]
  By IH, if one of these $A_2, \ldots, A_m,\langle G_j\rangle, [H_k]$'s is provable, then we are done; otherwise $\Rightarrow [\Rightarrow B]$ is provable, by ($\square_R$), we see $\Rightarrow \square B$, i.e. $\Rightarrow A_1$ is provable as well.
  
  Let $A_1=\Diamond B$, assume $l\geq 1$ and the modal block where $B$ is propagated is just $[H_1]$. Recall $H_1=\ \Rightarrow\Theta_1$ and then $S$ is derived by
    \[
    \AxiomC{$\Rightarrow \Diamond B,A_2, \ldots, A_m, \overline{\langle G_j\rangle},[H_2],\ldots,[H_l],[\Rightarrow\Theta_1,B]$}
    \RightLabel{($\Diamond_R$)}
    \UnaryInfC{$\Rightarrow \Diamond B,A_2, \ldots, A_m, \overline{\langle G_j\rangle},[H_2],\ldots,[H_l],[\Rightarrow\Theta_1]$}
    \DisplayProof
    \]
 By IH, if one of these $A_i,\langle G_j\rangle,[H_k]$'s where $i,j\geq 1$ and $k\geq 2$ is provable, then we are done; otherwise $[\Rightarrow\Theta_1,B]$ is provable. Denote $\Rightarrow\Theta_1,B$ by $T$, by definition, we see $md(T)< md(S)$, so we can apply IH to $T$. It follows that either $B$ or an object $\+O\in\Theta_1$ is provable. 
  If such a $\+O$ is provable, by weakening, we see $\Rightarrow\Theta_1$ is provable as well; otherwise $B$ is provable, also by weakening, it follows directly $\Rightarrow[\Rightarrow B]$ is provable. Then we can construct a derivation as follows
    \[
    \AxiomC{$\Rightarrow[\Rightarrow B]$}
    \RightLabel{($w_R$)}
    \UnaryInfC{$\Rightarrow[\Rightarrow B],\Diamond B$}
    \RightLabel{($\Diamond_R$)}
    \UnaryInfC{$\Rightarrow[\Rightarrow],\Diamond B$}
    \RightLabel{(\axiomd)}
    \UnaryInfC{$\Rightarrow \Diamond B$}
    \DisplayProof
    \]
  Thus $\Rightarrow\Diamond B$ is provable and we are done.
  
  Next, suppose $S$ is derived by applying a rule within some $\langle G_j\rangle$ or $[H_k]$. Left rule applications can only be applied to $\langle G_j\rangle$ and are trivial, as for right rule applications, the reasoning is the same as above. 

  Lastly, we consider $(\text{inter}_\rightarrow)$ which is applied between some $\langle G_j\rangle$ and $[H_k]$. We further assume such $\langle G_j\rangle$ and $[H_k]$ are just $\langle G_1\rangle$ and $[H_1]$ in order to simplify the index. 
  Let $G_1=\Sigma \Rightarrow\Pi$, recall $H_1=\ \Rightarrow\Theta_1$ and then $S$ is derived by
    \[
    \AxiomC{$\Rightarrow \overline{A_i}, \langle G_2\rangle,\ldots, \langle G_n\rangle,[H_2],\ldots,[H_l],\langle\Sigma\Rightarrow\Pi,[\Rightarrow\Theta_1^\flat]\rangle,[\Rightarrow\Theta_1]$}
    \RightLabel{$(\text{inter}_\rightarrow)$}
    \UnaryInfC{$\Rightarrow \overline{A_i}, \langle G_2\rangle,\ldots, \langle G_n\rangle,[H_2],\ldots,[H_l],\langle\Sigma\Rightarrow\Pi\rangle,[\Rightarrow\Theta_1]$}
    \DisplayProof
    \]
  By IH, if one of these $A_i,\langle G_j\rangle,[H_k]$'s where $i\geq 1$ and $j,k\geq 2$, or $[\Rightarrow\Theta_1]$ is provable, then we are done; otherwise $\Rightarrow \langle\Sigma\Rightarrow\Pi,[\Rightarrow\Theta_1^\flat]\rangle$ is provable. Recall for each sequent $T\in^+ H_1$, $T$ has an empty antecedent, so 
  $[\Rightarrow\Theta_1^\flat]$ is an empty structure. By Proposition \ref{prop:empty-structure}, we conclude $\Rightarrow \langle\Sigma\Rightarrow\Pi\rangle$ is provable in \cdid.

  For \cdit, the proof is almost the same as what we have done for \cdid, the only different sub-case is the one when $A_1=\Diamond B$ and in the premise $B$ is provable. In \cdit, we can construct a derivation as
  \[
    \AxiomC{$\Rightarrow B$}
    \RightLabel{($w_R$)}
    \UnaryInfC{$\Rightarrow B, \Diamond B$}
    \RightLabel{($\axiomt_\Diamond$)}
    \UnaryInfC{$\Rightarrow \Diamond B$}
    \DisplayProof
  \]
  Then we are done.
  \qed
\end{proof}

%% file: A-sec4-proofs.tex
\section{Proofs in Section \ref{sec:termination}}

{\noindent \textbf{Proposition~\ref{sequent-inclusion}.}} 
{\it Let $S=\Gamma\Rightarrow\Delta$. If $S$ is saturated with respect to $(\text{trans})$, $(\text{inter}_{\rightarrow})$ and ($\text{inter}_{\downarrow}$), then for $\langle\Sigma\Rightarrow\Pi\rangle\in\Delta$, we have $\Gamma\Rightarrow\Delta\subseteq^\-S\Sigma\Rightarrow\Pi$.}

\begin{proof}
    It follows directly from Definition \ref{def:saturation-inclusion} and saturation conditions for $(\text{trans})$, $(\text{inter}_{\rightarrow})$ and ($\text{inter}_{\downarrow}$).\qed
\end{proof}

{\noindent \textbf{Proposition~\ref{prop:realization}.}} 
    {\it Let $S=\Gamma\Rightarrow\Delta,\langle S_1\rangle,[S_2]$, where $S_1=\Sigma\Rightarrow\Pi,[\Lambda\Rightarrow\Theta]$ and $S_2= \ \Rightarrow_{\Lambda\Rightarrow\Theta}\Theta^\#$ and $\Gamma\subseteq\Sigma$. 
    If $S_1$ is saturated with respect to all the left rules in {\rm \textbf{C}}\cdikk, then for the sequent $S'=\Gamma\Rightarrow\Delta,\langle S_1\rangle,[f_{S_1}(S_2)]$ which is obtained by the realization procedure in Definition \ref{def:realization}, we have
    \begin{enumerate}
      \item[(i).] $S'$ is saturated with respect to all the left rules applied to or within $[f_{S_1}(S_2)]$;
      \item[(ii).] $f_{S_1}(S_2)\subseteq^\-S \Lambda\Rightarrow\Theta$;
      \item[(iii).] $S'$ can be obtained by applying left rules of  {\rm \textbf{C}}\cdikk~to $[S_2]$ in $S$.
    \end{enumerate}
    }

\begin{proof}
    We abbreviate $f_{S_1}(S_2)$ as $f(S_2)$ in the proof.
    \begin{enumerate}
      \item[(i).] We only show $S'$ is saturated with $(\Diamond_L)$ and $(\square_L)$ on the block $[f(S_2)]$. If $\square A\in\Gamma$, then by saturation condition, we have $A\in\Theta$. By definition, we see that $A\in Fm(S_2)$ which is inherited in the succedent of $f(S_2)$ as well. If $\Diamond A\in\Gamma$ and $A\in\Lambda$, in this case, by the definition of tracking record, we see $A\in\mathfrak{G}(\Lambda\Rightarrow\Theta,\Gamma)$, which is the antecedent of $f(S_2)$. 
      Left rule saturations for formulas deeply within the block $[f(S_2)]$ can be verified similarly. 
      \item[(ii).] 
      Note that for $T\in^+ S_2$, $T$ is of the form $\Rightarrow_{\Phi\Rightarrow\Psi}\Psi^\#$ for some $\Phi\Rightarrow\Psi\in^+\Lambda\Rightarrow\Theta$. 
      We claim that $f(T)\subseteq^\-S \Phi\Rightarrow\Psi$. We show this by induction on the structure of $T$. If $T$ is block-free, by definition, $Ant(f(T))=\mathfrak{G}(\Phi\Rightarrow\Psi,\Gamma)\subseteq Ant(\Phi\Rightarrow\Psi)$ and then we are done. 
      Otherwise, for each $T'\in^{[\cdot]}T$ and $T'=\ \Rightarrow_{\Phi'\Rightarrow\Psi'}\Psi'^\#$, By IH, $f(T')\subseteq^\-S \Phi'\Rightarrow\Psi'$. 
      The other direction can be verified similarly. Thus $f(T)\subseteq^\-S \Phi\Rightarrow\Psi$. 
      
      \noindent Since $S_2\in^+S_2$, as a result, we have $f(S_2)\subseteq^\-S \Lambda\Rightarrow\Theta$.
      \item[(iii).] Directly from (i).
      \qed
    \end{enumerate}
\end{proof}

\begin{lemma}\label{lem:boundimp}
	Given a sequent $S_0$, let $\+D =  PROC_0(S_0)$,   $S$  be a global-R2-saturated sequent occurring in $\+D$,  $T\in^+ S$,  and  $\Omega_T = \{T' \in^{\langle\cdot\rangle} T\}$, then $\Omega_T$  is finite and bounded by $O(|S_0|^{|S_0|})$.
\end{lemma}

\begin{proof}
    Let $S$ and  $T\in^+ S$ as in the hypothesis,  and let us call any $T^* \in_0^{\langle\cdot\rangle} T$ a $\langle\cdot\rangle$-child of $T$. Since each $\langle\cdot\rangle$-child of $T$ is generated  by some formula  $C\supset D\in Sub(S_0)$,  $T$ has at most $O(|S_0|)$  $\langle\cdot\rangle$-children. 
    It suffices to show that for any  sequent $S'=S$ or $S'$ occurring above $S$ in $\+D$, every  $\in_0^{\langle\cdot\rangle} $-chain of  implication blocks starting   from $T$ in $S'$ has a length  bounded by $O(|S_0|)$. 
    To this purpose let $\ldots T_{i+1} \in_0^{\langle\cdot\rangle} T_i \in_0^{\langle\cdot\rangle} \ldots \in_0^{\langle\cdot\rangle} T_0 = T$ be such a chain. Observe that each $T_i$ is Global R2-saturated by definition of $PROC_0$. For any $i$, $T_{i+1}$ is a  $\langle\cdot\rangle$-child of $T_i$, thus it must be:  
$T_i = \Sigma_i \Rightarrow \Pi_i, C \supset D, \langle T_{i+1}\rangle$, where: 
    \begin{quote}
    - $C \supset D\in Sub(S_0)$\\
    - $C\not\in \Sigma_i$\\
    - $T_{i+1} = \Sigma_{i+1} \Rightarrow \Pi_{i+1}$\\
    - $\Sigma_i \subseteq \Sigma_{i+1}$\\
    - (*) $C\in \Sigma_{i+1}$\\
    - $D\in \Pi_{i+1}$
    \end{quote}
    Observe that for  any  $T_j$ in the chain with $j \geq i+1$,  and  $T_j= \Sigma_j \Rightarrow \Pi_j$,   even if $C \supset D  \in \Pi_j$, the formula $C \supset D$  \emph{cannot} be used to generate again a  $\langle\cdot\rangle$-child of
    $T_j$ because of: (*), the fact $\Sigma_{i+1}\subseteq \Sigma_j$ (since $T_j$ is R2- saturated) and the modified implication rule. Thus the set of implication subformulas  of $S_0$ that can generate   implication blocks 
    \emph{strictly decreases} along the chain, and we may conclude that the chain  has length $O(|S_0|)$.
    \qed
\end{proof}

{\noindent \textbf{Proposition~\ref{proc0term}.}} 
{\it Given a sequent $S_0$, ${\rm PROC}_0(S_0)$ produces a finite derivation with all the leaves axiomatic or at least one global-R3-saturated leaf.}

\begin{proof}
	(Sketch) We prove first that $\text{PROC}_0(S_0)$ terminates by producing a finite derivation. Then either  all leaves are axiomatic, or there must be at least one leaf that is {\rm global-R$3$-saturated}, otherwise  one of the leaves $S$ would be selected in Step 8 and  further expanded,   contradicting the fact that is indeed a leaf.  Thus it suffices to  prove that the procedure produces a finite  derivation. 
    Let  $\+D$ be the derivation built by $\text{PROC}_0(S_0)$.  First we claim that all the branches of $\+D$ are finite. 
    Suppose for the sake of a contradiction that $\+D$ contains an infinite branch ${\cal B} = S_0,\ldots, S_i, \ldots $. The branch is generated by applying repeatedly $\textbf{EXP1}(\cdot), \textbf{EXP2}(\cdot)$ and $\textbf{EXP3}(\cdot)$ to each   $S_i$. 
    Since each of these sub-procedures terminates, the three of them must infinitely alternate on the branch. It is easy to see that if  $T_i\in^+ S_i$ satisfies a saturation condition for a rule ($r$), it will still satisfy it in all $S_j$ with $j >i$. 
    We can conclude that the branch must contain infinitely many phases of $\textbf{EXP3}(\cdot)$ each time generating  new implication blocks; therefore sequents  $S_i\in {\cal B}$ will eventually contain an \emph{unbounded} number of implication blocks, contradicting the previous Lemma \ref{lem:boundimp}.
	Thus each branch of the derivation $\+D$ built by $\text{PROC}_0(S_0)$ is finite. 
    To conclude the proof, observe that $\+D$ is a tree whose branches have a finite length and is finitely branching (namely each node/sequent has at most two successors, as the rules in \textbf{C}\cdikk~are at most binary), therefore $\+D$ is finite.
	\qed
\end{proof}

\begin{proposition}\label{prop:md-preserve}
    Let $(r)$ be a rule in {\rm \textbf{C}}\cdikk~and of one of the following forms,
    \[
        \AxiomC{$S_0$}
        \RightLabel{($r$)}
        \UnaryInfC{$S$}
        \DisplayProof
        \quad
        \AxiomC{$S_1$}
        \AxiomC{$S_2$}
        \RightLabel{($r$)}
        \BinaryInfC{$S$}
        \DisplayProof
      \]
    If $md(S)=k$, then in both cases, the modal degree of the premise(s) remains the same, i.e. $md(S_0)=md(S_1)=md(S_2)=k$.
\end{proposition}
  
\begin{proof}
    By checking all the rules in \textbf{C}\cdikk~one by one.\qed
\end{proof}

{\noindent \textbf{Theorem~\ref{termination}.}} 
{\it Let $A$ be a formula. Proof-search for $\Rightarrow A$ in {\rm \textbf{C}}\cdikk~terminates with a finite derivation in which either all the leaves are axiomatic or there is at least one global-saturated leaf.}

\begin{proof}
    By Proposition \ref{proc0term}, $\text{PROC}_0(\cdot)$ terminates. 
    Also, by Proposition \ref{prop:expi-terminates}, each phases of $\textbf{EXP4}(\cdot)$ terminates. Hence 
    it suffices to show that in the whole process of $\text{PROC}(A)$, $\text{PROC}_0(\cdot)$ can only be iterated for finitely many rounds.
  
    Assume $md(A)=k$. The proof-search procedure is initialized by $\text{PROC}_0(\Rightarrow A)$, if $A$ is provable, then the proof is built by $\text{PROC}_0(\Rightarrow A)$, otherwise we obtain a finite derivation $\+D$ with at least one global-R3-saturated leaf. 
    Take such a global-R3-saturated leaf $S$ from $\+D$. Then following the algorithm, first $(\text{inter}_{\downarrow})$ is applied to $S$, we see that 
    $S$ is expanded to some $S'$ where each implication block of the form $\langle\Sigma\Rightarrow\Pi,[\Lambda_0\Rightarrow\Theta_0]\rangle\in^+S$ is expanded to $\langle\Sigma\Rightarrow\Pi,[\Lambda_0\Rightarrow\Theta_0]\rangle,[\Rightarrow\Theta_0^\#]$. After the realization procedure, the previously produced modal block $[\Rightarrow\Theta_0^\#]$ is expanded to some $[\Lambda_1\Rightarrow\Theta_1]$. Now in order to make the whole $S'$ global-R3-saturated, we turn to the previous step $\text{PROC}_0(\cdot)$ again. In this case, recall that $S$ is already global-R3-saturated and the only difference between $S$ and $S'$ is such blocks like $[\Lambda_1\Rightarrow\Theta_1]$, namely the modal blocks produced in the previous $\textbf{EXP4}(\cdot)$ procedure. 
    
    In the next round of $\text{PROC}_0(\cdot)$, rules from R1 and R2 are first applied to these $[\Lambda_1\Rightarrow\Theta_1]$-blocks and we obtain $[\Lambda_2\Rightarrow\Theta_2]$. Now when we turn to R3-rule, namely $(\supset_R)$, it can only be applied to $\Lambda_2\Rightarrow\Theta_2$ itself. We claim for such a sequent $\Lambda_2\Rightarrow\Theta_2$, $md(\Lambda_2\Rightarrow\Theta_2)<md(S)$. 
    Assume $[\Lambda_2\Rightarrow\Theta_2]\in^+S''$ where $S''$ is the expansion of $S$ after all the previous procedures before we obtain $\Lambda_2\Rightarrow\Theta_2$. 
    Note that $S$ is obtained by applying R1-R3 rules to $\Rightarrow A$, by Proposition \ref{prop:md-preserve}, we have $md(S)=md(A)=k$. Next, consider the steps expanding $S$ to $S''$. By Proposition \ref{prop:md-preserve}, backward $(\text{inter}_{\downarrow})$ application also preserves modal degree of a sequent. While for the realization procedures, according to Proposition \ref{prop:realization} (iii), the procedures can be simulated by left rule applications in \textbf{C}\cdikk, which by Proposition \ref{prop:md-preserve}, preserve modal degree. Therefore, we conclude $md(S'')=md(S)=k$ as well. Since $[\Lambda_2\Rightarrow\Theta_2]$ is a block in $S''$, by definition, we see that $md(\Lambda_2\Rightarrow\Theta_2)\leq k-1$ which is strictly smaller than $md(S)$. 
    This means each time we start a new round of $\text{PROC}_0(\cdot)$ in the loop, when we go to the $\textbf{EXP3}(\cdot)$ step, the modal degree of the sequent(s) to which R3 rule applies goes down strictly. 
    Thus $\textbf{EXP3}(\cdot)$ and also $\text{PROC}_0(\cdot)$ can only be iterated for finitely many rounds.
    \qed
\end{proof}

%% file: A-sec5-proofs.tex
\section{Proofs in Section \ref{sec:completeness}}

For convenience, we abbreviate $x_{\Phi\Rightarrow\Psi},\leq_S,R_S,W_S$ as $x,\leq, R,W$ respectively in the following proofs.\\

{\noindent \textbf{Proposition~\ref{prop:hp-fc-property}.}} {\it $\+M_S$ satisfies (FC) and (DC).}

\begin{proof}
    We only show the case of (DC) here, for (FC), the proof is similar with \cite[Proposition 51]{FIK-csl2024}. 
    Take arbitrary $x_{\Gamma\Rightarrow\Delta},x_{\Sigma\Rightarrow\Pi},x_{\Lambda\Rightarrow\Theta}\in W_S$ with $x_{\Gamma\Rightarrow\Delta}\leq x_{\Sigma\Rightarrow\Pi}$ and $Rx_{\Sigma\Rightarrow\Pi}x_{\Lambda\Rightarrow\Theta}$, our goal is to find some $x_0\in W_S$ s.t. both $x_0\leq x_{\Lambda\Rightarrow\Theta}$ and $Rx_{\Gamma\Rightarrow\Delta}x_0$ hold. 
    Since $Rx_{\Sigma\Rightarrow\Pi}x_{\Lambda\Rightarrow\Theta}$, by the definition of $R$, we see that $[\Lambda\Rightarrow\Theta]\in\Pi$. Meanwhile, since $x_{\Gamma\Rightarrow\Delta}\leq x_{\Sigma\Rightarrow\Pi}$, by the definition of $\leq$, we have $\Gamma\Rightarrow\Delta\subseteq^\-S\Sigma\Rightarrow\Pi$. 
    Given that $[\Lambda\Rightarrow\Theta]\in\Pi$, by Definition \ref{def:saturation-inclusion}, it implies there is a block $[\Phi\Rightarrow\Psi]\in\Delta$ s.t. $\Phi\Rightarrow\Psi\subseteq^\-S\Lambda\Rightarrow\Theta$. 
    In the meantime, $[\Phi\Rightarrow\Psi]\in\Delta$, it follows that $Rx_{\Gamma\Rightarrow\Delta}x_{\Phi\Rightarrow\Psi}$.\qed
\end{proof}
  
{\noindent \textbf{Lemma~\ref{lem:truth-lemma-cdik}.}} 
{\it Let $S$ be a global-saturated sequent in {\rm \textbf{C}}\cdikk~and $\+M_{S}=(W_S,$
$\leq_S,R_S,V_S)$ defined as above. (a). If $A\in\Phi$, then $\+M_{S},x_{\Phi\Rightarrow\Psi}\Vdash A$; (b). If $A\in\Psi$, then $M_{S},x_{\Phi\Rightarrow\Psi}\nVdash A$.}

\begin{proof}
    We prove the lemma by induction on $A$ and only present the non-trivial case when $A$ is of the form $\square B$.  
  
    For (a), let $\square B\in\Phi$. We see $\Phi\Rightarrow\Psi$ satisfies the saturation condition associated with $(\square_L)$ for $\square B$. Assume for the sake of a contradiction that $x\nVdash \square B$. 
    Then there exists $x_{\Sigma\Rightarrow\Pi}\in W$ s.t. $Rxx_{\Sigma\Rightarrow\Pi}$ and $x_{\Sigma\Rightarrow\Pi}\nVdash B$. By IH, we see that $B\notin\Sigma$. Meanwhile, since $Rxx_{\Sigma\Rightarrow\Pi}$, according to the model construction, we have $[\Sigma\Rightarrow\Pi]\in\Psi$. By the saturation condition associated with $(\square_L)$, we have $B\in \Lambda$, which leads to a contradiction. 
    For (b), let $\square B\in \Psi$. 
    We see $\Phi\Rightarrow\Psi$ satisfies the saturation condition associated with $(\square_{R'})$ for $\square B$. By the saturation condition, there is a block $[\Lambda\Rightarrow\Theta]\in\Psi$ with $B\in\Theta$. By IH, we have $x_{\Lambda\Rightarrow\Theta} \nVdash B$. According to the model construction, we have $Rxx_{\Lambda\Rightarrow\Theta}$, so $x \nVdash \Box B$. \qed
\end{proof}

\begin{lemma}[Truth lemma for {\rm \textbf{C}}\cdid~and {\rm \textbf{C}}\cdit]\label{lem:truth-lemma-cdid-cdit}
  Let $S$ be a global-saturated sequent in {\rm \textbf{C}}\cdid (resp. {\rm \textbf{C}}\cdit) and $\+M_{S}=(W_S,\leq_S,R_S,V_S)$ defined as above. 
  (a). If $A\in\Phi$, then $\+M_{S},x_{\Phi\Rightarrow\Psi}\Vdash A$; (b). If $A\in\Psi$, then $M_{S},x_{\Phi\Rightarrow\Psi}\nVdash A$.
\end{lemma}

\begin{proof}
    The proof is done by induction on $A$, the only two non-trivial cases are $A=\square B\in\Phi$ and $A=\Diamond B\in\Psi$. 
    \begin{enumerate}
      \item[(i).] Let $S$ be a global-saturated sequent in {\rm \textbf{C}}\cdid.
      \begin{enumerate}
        \item[(a).] Let $A=\square B\in\Phi$. Assume for the sake of a contradiction that $x_{\Phi\Rightarrow\Psi}\nVdash \square B$, then there is $x_{\Lambda\Rightarrow\Theta}\in W$ s.t. $Rxx_{\Lambda\Rightarrow\Theta}$ and $x_{\Lambda\Rightarrow\Theta}\nVdash B$. By IH, we have $B\notin\Lambda$. 
        Since $Rxx_{\Lambda\Rightarrow\Theta}$, by the model construction, either $[\Lambda\Rightarrow\Theta]\in\Psi$ or $\Lambda\Rightarrow\Theta=\Phi\Rightarrow\Psi$ and $\Psi$ is $[\cdot]$-free. 
        Note that $\Phi^\square\cup\Psi^\Diamond$ is non-empty, then by the saturation condition associated with (\axiomd), $\Psi$ is not $[\cdot]$-free, so $\Lambda\Rightarrow\Theta$ cannot be $\Phi\Rightarrow\Psi$ itself, which entails $[\Lambda\Rightarrow\Theta]\in\Psi$. By the saturation condition associated with $(\square_R)$, we have $B\in\Lambda$, a contradiction. 
        \item[(b).] Let $A=\Diamond B\in\Psi$. In this case, $\Phi^\square\cup\Psi^\Diamond$ is non-empty, then by the saturation condition associated with (\axiomd), we see that there is a block $[\Lambda\Rightarrow\Theta]\in\Psi$. Furthermore, by the saturation condition associated with $(\Diamond_R)$, we have $B\in\Theta$. By IH, $x\nVdash B$. Meanwhile, according to the model construction, we have $Rxx_{\Lambda\Rightarrow\Theta}$, and hence $x\nVdash\Diamond B$.
      \end{enumerate}
      \item[(ii).] Let $S$ be a global-saturated sequent in {\rm \textbf{C}}\cdit.
      \begin{enumerate}
        \item[(a.)] Let $A=\square B\in\Phi$. 
        Assume for the sake of a contradiction that $x\nVdash \square B$, then similarly as in \textbf{C}\cdid, there is $x_{\Lambda\Rightarrow\Theta}\in W$ s.t. $Rxx_{\Lambda\Rightarrow\Theta}$ and $B\notin\Lambda$. By the model construction, either $[\Lambda\Rightarrow\Theta]\in\Psi$ or $\Lambda\Rightarrow\Theta=\Phi\Rightarrow\Psi$. For the former, by the saturation condition associated with $(\square_R)$, we have $B\in\Lambda$, a contradiction. For the latter, since $\Phi\Rightarrow\Psi$ is saturated with $(\axiomt_\square)$, we see that $B\in\Phi$, i.e. $B\in\Lambda$ as well, also a contradiction. 
        \item[(b).] Let $A=\Diamond B\in\Psi$. Since $\Phi\Rightarrow\Psi$ is saturated with $(\axiomt_\Diamond)$, we see that $B\in\Psi$ as well. By IH, $x\nVdash B$. By the model construction, $R$ is reflexive, so we have $Rxx$, which makes $x\nVdash\Diamond B$ neither.\qed
      \end{enumerate}
    \end{enumerate}
\end{proof}